\documentclass[11pt,a4paper,twoside,reqno]{amsart}
\renewcommand\thesection{\arabic{section}}
\renewcommand\thesubsection{\thesection.\arabic{subsection}}

\makeatletter
\def\@settitle{\begin{center}%
		\baselineskip14\p@\relax
		\normalfont\LARGE\scshape\bfseries
		\@title
	\end{center}%
}
\makeatother

\makeatletter

\def\subsection{\@startsection{subsection}{2}%
	\z@{.5\linespacing\@plus.7\linespacing}{.5\linespacing}%
	{\normalfont\large\bfseries}}

\makeatother

\makeatother

\usepackage[usenames, dvipsnames]{color}
\definecolor{darkblue}{rgb}{0.0, 0.0, 0.45}

\usepackage[colorlinks	= true,
raiselinks	= true,
linkcolor	= darkblue, 
citecolor	= Mahogany,
urlcolor	= ForestGreen,
pdfauthor	= {Peyman Mohajerin Esfahani},
pdftitle	= {},
pdfkeywords	= {},
pdfsubject	= {},
plainpages	= false]{hyperref}

\usepackage{dsfont,amssymb,amsmath,subfigure, graphicx,enumitem,multirow} 
\usepackage{amsfonts,dsfont,mathtools, mathrsfs,amsthm} 
\usepackage[amssymb, thickqspace]{SIunits}
\usepackage{algorithm,fancyhdr}
\usepackage{algorithmicx}
\usepackage{algpseudocode}

\allowdisplaybreaks
\date{\today}

\newcommand{\rme}{{\rm e}}
\newcommand{\mPr}{\mathbf{Pr}}
\newcommand{\mE}{\mathbf{E}}
\newcommand{\mfq}{\mathfrak{q}}

\newcommand{\Fij}{\mathds{F}_{ij}}
\newcommand{\Fjj}{\mathds{F}_{jj}}

\newcommand{\Sij}{\mathsf{S}_{ij}}
\newcommand{\Sih}{\mathsf{S}_{ih}}
\newcommand{\hsigma}{\hat{\sigma}}

\newcommand{\T}{\mathds{T}}
\newcommand{\F}{\mathds{F}}
\newcommand{\N}{\mathbb{N}}
\newcommand{\Ht}{\mathcal{H}_2}

\newcommand{\X}{{\mathcal{X}}}

\newcommand{\Nmodeset}{\mathcal{N}}

\newcommand{\R}{\mathbb{R}}

\newcommand{\note}[1]{\textcolor{red}{(#1)}}

\newcommand{\B}{\mathbb{B}}

\newcommand{\EZ}[1][]{e}
\newcommand{\fEZ}[1][]{e}

\usepackage{epsfig} 
\usepackage{graphicx}
\usepackage{grffile}
\usepackage{amsmath}
\usepackage{enumitem}
\usepackage{booktabs}
\usepackage{multirow}
\usepackage{array}
\usepackage{cite}
\usepackage{bm}
\usepackage{pdfpages} 
\usepackage{subfigure}
\usepackage{booktabs}
\usepackage{caption}
\usepackage{dsfont}
\newtheorem{theorem}{Theorem}[section]
\newtheorem{lemma}[theorem]{Lemma}

\newtheorem{proposition}[theorem]{Proposition}

\newtheorem{remark}[theorem]{Remark}
\newtheorem{assumption}[theorem]{Assumption}
\numberwithin{equation}{section}   

\begin{document}

\title{Multimode Diagnosis for Switched Affine Systems with Noisy Measurement}

\author{Jingwei Dong, Arman Sharifi Kolarijani and Peyman Mohajerin Esfahani}
\thanks{The authors are with the Delft Center for Systems and Control, Delft University of Technology, The Netherlands (\{J.Dong-6, P.MohajerinEsfahani\}@tudelft.nl, Arman.Sh.Kolarigani@gmail.com). This work is supported by the ERC grant TRUST-949796 and CSC (China Scholarship Council) with funding number: 201806120015.}
\maketitle

\begin{abstract}
We study a diagnosis scheme to reliably detect the active mode of discrete-time, switched affine systems in the presence of measurement noise and asynchronous switching. The proposed scheme consists of two parts: (i) the construction of a bank of filters, and (ii) the introduction of a residual/threshold-based diagnosis rule. 
We develop an exact finite optimization-based framework to numerically solve an optimal bank of filters in which the contribution of measurement noise to the residual is minimized. The design problem is safely approximated through linear matrix inequalities and thus becomes tractable.
We further propose a thresholding policy along with probabilistic false-alarm guarantees to estimate the active system mode in real-time. In comparison with the existing results, the guarantees improve from a polynomial dependency in the probability of false alarm to a logarithmic form. This improvement is achieved under the additional assumption of sub-Gaussianity, which is expected in many applications. The performance of the proposed approach is validated through a numerical example and an application of the building radiant system. 
\end{abstract}

\section{Introduction}
Over the last two decades, special attention has been paid to switched affine systems because they can be used to effectively model a wide range of practical systems, such as chemical plants~\cite{venkatasubramanian2003review}, aeronautic systems~\cite{zolghadri2012advanced} and smart buildings~\cite{weimer2013parameter}. 
These systems are usually difficult to be exactly described by a single model because of their nonlinear and complex dynamic characteristics. 
Research on switched systems is mainly focused on model identification~\cite{bako2011identification,ohlsson2013identification}, state estimation~\cite{ackerson1970state}, stability analysis and controller design~\cite{lin2009stability,yuan2018novel}.  
The prior knowledge of the switching signal that indicates the evolution of modes is crucial to theoretical results in these research topics.  
For example, a general approach to controlling switched systems is to employ mode-dependent controllers, where the activation of a proper controller depends on the switching signal. There are, however, several scenarios in which the switching signal is not a priori known. In fault diagnosis scenarios, an unexpected transition from a healthy mode to a faulty mode can be treated as an unknown switching. Thus, one needs to detect the active mode of switched systems as the detection process results in a delay between the active mode and its corresponding controller.

\subsection{Literature review}
The problem of mode detection for switched affine systems has been studied for decades. 
The proposed approaches can be grouped into two categories: data-based and model-based approaches.
The data-based approaches are most adopted when the parameters of each mode are unknown. In that case, the parameters need to be identified from a collection of input-output data. Then the new data is associated with the most suitable mode through data classification techniques.  
A number of results on data-based approaches have been achieved. We refer the interested readers to~\cite{bako2011identification} and the references therein.    

{\em Model-based fault diagnosis:} In model-based approaches, one leverages tools from the fault detection and isolation (FDI) field to detect and isolate changes caused by switches or faults.
The most widely used FDI methods are based on residual generation, where certain residual signals are generated by observer-based or parity space methods to characterize the occurrence of changes quantitatively~\cite{gao2015survey}. 
Beard~\cite{beard1971failure} proposes the original observer-based diagnosis approach to replace the hardware redundancy in 1971. Subsequently, many observer-based diagnosis approaches are developed. To deal with disturbances or measurement noise, the authors in~\cite{henry2005design} construct an optimization problem to design the parameters of the observer, in which the influence of disturbances on residuals characterized by~$\mathcal{H}_{\infty}$-norm is minimized.
The parity space approach is proposed in~\cite{chow1984analytical}, which generates residuals to check the consistency between the model and the measurements. 
It is worth noting that the derived residual generators usually have the same order as that of the systems. This makes the generators complex and computationally demanding when dealing with high-dimensional or large-scale systems.
Frisk~\cite{frisk2001minimal} proposes a parity-space-like approach in a polynomial framework which produces residual generators with possibly low order. In their following work~\cite{nyberg2006residual}, the previous approach is extended to the linear differential-algebraic equation (DAE, difference-algebraic equation in the discrete-time case). This extension enlarges the application range of FDI approaches because DAE models cover several classes of models, e.g., transfer functions, state-space models, or descriptor models.
The above approaches are for linear systems. 
For the fault detection of nonlinear systems, a natural way is to linearize nonlinear systems at local operating points and decouple the disturbances together with the higher-order terms from the residuals, see for example~\cite{seliger1991fault,benosman2010survey}. Another method is to develop adaptive nonlinear estimators to approximate the nonlinear terms \cite{boem2011distributed,ferrari2011distributed}.
More recently, the authors in~\cite{esfahani2015tractable,pan2021dynamic} develop tractable optimization-based approaches in the DAE framework to design FDI filters to deal with disturbances and nonlinear terms. 

{\em Multi-mode diagnosis:} Note that the aforementioned approaches are applicable to systems with a single model. A bank of residuals is usually required to deal with systems consisting of several modes. Moreover, the systems need to satisfy certain rank conditions to guarantee that any two subsystems can be distinguished from each other. This is the distinguishability (also called discernibility or observability) of switched systems~\cite{halimi2014model,kusters2018switch}.
To detect the active mode, the idea that makes each residual sensitive to all but only one mode is usually adopted, which is called \textit{generalized observer scheme} (GOS)~\cite{frank1990fault}. Following a GOS mindset, results on mode detection are achieved based on basic residual generation methods, such as parity space approaches~\cite{cocquempot2004fault}, unknown input observers~\cite{wang2007adaptive}, and sliding mode observers~\cite{mincarelli2016uniformly,zhang2019sliding}. Note that the computational complexity of these residual generation methods increases significantly as the system dimension and the number of modes increase. In this work, we propose a design perspective in the DAE framework that relies on a bank of filters whose dimension does not necessarily scale up with the dimension of the system. This feature enables a possibility of low-ordered filters compared to the existing literature. 

Another class of mode detection methods is the set-membership method which computes the reachable set of each subsystem. 
Then, the output is compared to the reachable sets to determine the mode~\cite{scott2014input,marseglia2017active, harirchi2018guaranteed}.
The authors in~\cite{scott2014input} and~\cite{marseglia2017active} develop active diagnosis approaches in which an optimal separating input sequence is designed to guarantee that output sets of different subsystems are separated.
In~\cite{harirchi2018guaranteed}, a model invalidation approach is proposed to compare the input-output data to the nominal behaviors of the system, where the set-membership check is reduced to the feasibility of a mixed-integer linear programming problem.
The set-membership methods are generally computationally demanding because they require solving optimization problems at each step. 
Also, the residual generation and set-membership methods mentioned above either neglect the noise or treat them as robust only through the support information. This viewpoint often leads to conservative diagnosis guarantees. In fact, the measurement noise introduces a unique challenge to the detection task where the reachable sets of healthy residuals may well overlap with the faulty ones. This challenge is one of the focus points of this study.

\subsection{Main contributions}
In the light of the literature reviewed above, the main message of this paper revolves around a diagnosis scheme to detect the active mode of asynchronously switched affine systems in real-time. The diagnosis scheme consists of a bank of filters and a residual/threshold-based diagnosis rule. The bank of filters comprises as many filters as the admissible mode transitions, while the diagnosis rule prescribes conditions under which we estimate the transition based on the behaviors of the residuals. The main contributions of this paper are summarized as follows.
\begin{itemize}
\item {\bf Exact characterization of an optimal bank of filters:}
Building on residual-based detection and $\Ht$-norm approaches in the DAE framework, we formulate the optimal bank of filters design problem as a finite optimization problem in which the objective is the noise contribution to the residuals~(Theorem~\ref{thm:minimization problem}). We also provide necessary and sufficient conditions that ensure the feasibility of the resulting optimization problem~(Proposition~\ref{prop:feasibility}). 

\item {\bf Tractable convex restriction:}
We provide an LMI-based sufficient condition for the nonlinear constraint in the exact optimization problem of the filters design, leading to a tractable approximation of the original problem~(Proposition~\ref{prop:LMI Appro}). 

\item {\bf Probabilistic performance bounds:}
We further propose diagnosis thresholds along with probabilistic false-alarm guarantees to estimate the active system mode~(Theorem~\ref{thm:Prob certificate}). The proposed bound admits a logarithmic dependency with respect to the desired reliability level, which is better than the polynomial rate in the existing works~\cite{boem2018plug}. This improvement comes under the sub-Gaussianity assumption on the noise distribution, a regularity requirement expected to hold in many real-world applications.
\end{itemize}

The rest of the paper is organized as follows. The problem formulation and the proposed architecture of the diagnosis scheme are introduced in Section~\ref{sec:problem description}. In Section~\ref{sec:main results}, we present an optimization-based approach to design the filters along with some performance analysis of the proposed scheme. 
To improve the flow of the paper and its accessibility, we postpone all technical proofs to Section~\ref{sec:proofs}.
The proposed scheme is applied to a numerical example and a building radiant system in Section~\ref{sec:simulation} to validate its effectiveness.
Finally, Section~\ref{sec:conslusion} concludes the paper with some remarks and future directions.

\paragraph{\bf Notation} 
Sets~$\R~(\R_{+})$ and~$\N~(\N_{+})$ denote all reals (positive reals) and non-negative (positive) integers, resp. 
Set~$\{1,\dots,n\}$ is denoted by~$\Nmodeset$.
Sets of symmetric matrices and non-singular matrices are denoted by~$\mathcal{S}$ and~$\mathcal{M}$, resp.
In symmetric matrices, we use~$\ast$ for the off-diagonal elements in an attempt to avoid clutter. The identity matrix with an appropriate dimension is denoted by~$I$.
The maximum singular value of a matrix~$A$ is denoted by~$\Vert A \Vert_2$.
For a vector~$v=[v_1,\dots,v_n]$, the~$2$-norm and infinity-norm of~$v$ are~$\|v\|_2 = \sqrt{\sum^n_{i=1} v_i^2}$ and~$\|v\|_\infty = \max_{i\in\{1,\dots,n\}}|v_i|$, resp.
For a random variable~$\chi$, the probability law and expected value are denoted by~$\mPr[\chi]$ and~$\mE[\chi]$, resp.
Given a signal~$s = \{s(k)\}_{k \in \N}$ and a LTI system (or transfer function)~$\T$, the notation $\T [s]$ denotes the output of the system in response to the input signal $s$. The composition of two transfer functions $\T$ and $\F$ is also denoted by $\F\T[s] =\F [\T[s]]$. We use the shorthand notation~$\| \T \|_{\Ht}$ to denote the~$\Ht$-norm of~$\T$. The steady-state gain of $\T$ is denoted by $[\T]_{{\rm ss}} := \lim_{k \to \infty}\T[1](k)$, whenever the limit exists.

\section{Model Description and Problem Statement}\label{sec:problem description}
In this section, a formal description of discrete-time asynchronously switched affine systems is given.  Then we present the architecture of the proposed mode detector and formulate the problems studied in this work. 

\subsection{Model description}
Consider a discrete-time switched affine system that consists of~$n$ subsystems:
\begin{align}\label{eq:system}
	x(k+1) &= A_{\sigma(k)} x(k) + B_{\sigma(k)} u(k) + E_{\sigma(k)} d(k) + W_{\sigma(k)} \omega(k), \notag\\
	y(k)   &= C_{\sigma(k)} x(k) + D_{\sigma(k)} \omega(k)
\end{align}
where~$x(k) \in \R^{n_x}$,~$u(k) \in \R^{n_u}$ and~$y(k) \in \R^{n_y}$ are the state, control input and output, resp. 
The signal~$d(k) \in \R^{n_d}$ and $\omega(k) \in \R^{n_\omega}$ represent the reference and noise signals, resp. For simplicity of analysis, we consider a one-dimensional reference signal, i.e.,~$n_d = 1$. Throughout this study, the noise~$\omega(k)$ is assumed to be independent and identically distributed (iid). The switching law~$\sigma(k) \in \Nmodeset$ indicates the active mode at each instant~$k$. Matrices~$A_{\sigma(k)},~B_{\sigma(k)},~E_{\sigma(k)},~W_{\sigma(k)},~C_{\sigma(k)}$ and~$D_{\sigma(k)}$ are all known with appropriate dimensions, and matrices~$E_i \neq 0$. For each mode~$i \in \Nmodeset$, we consider the static output-feedback controller 
\begin{align}\label{eq:controller}
    u(k) = K_i y(k),
\end{align}
where~$K_i$ is a constant controller gain; see \cite{chang2013new} for a design approach to~$K_i$. Let~$\{t_0,t_1,\dots,t_s,\dots \}$ denote the sequence of switching time instants of the system mode~$\sigma(k)$, i.e., by definition we have $\sigma(t_s-1) \neq \sigma(t_s)$.

In this study, we consider the setting that the switching law~$\sigma(k)$ and the switching instant~$t_s$ are both unknown to the controller. The main objective is to estimate the active mode~$\sigma(k)$, hereafter denoted by~$\hsigma(k)$, through the noisy measurement~$y$ in real-time. As depicted in Figure~\ref{fig:closed-loop with FDI}, our proposed scheme to accomplish this goal builds on a bank of filters where each of which is intended to detect a possible pair of~$\hsigma(k) = i, \sigma(k) = j$ for any~$i,j\in \Nmodeset$; we use the notation~$\Sij$ to represent this status of the closed-loop system. For each pair~$(i,j)$, the filter is assumed to be a linear time-invariant (LTI) system (or transfer function) denoted by~$\Fij$ whose output (also called residual) is a scalar-valued signal~$r_{ij} := \Fij [y]$. We note that in our setting, the current controller mode~$\hsigma(k)=i$ is always known, whereas the system mode~$\sigma(k)$ is unknown and the object of interest. Suppose that the system transitions to the status~$\Sih$ at~$t_s$ (i.e.,~$\hsigma(t_s) = i, \sigma(t_s) = h$), thanks to the linearity of the dynamics, the residual $r_{ij}$ can be written as \begin{align}\label{eq:Design form Fij}
    r_{ij} =  \underbrace{\Fij\T^{\Sih}_{dy}}_{d \mapsto r_{ij}}[d] + \underbrace{\Fij\T^{\Sih}_{\omega y}}_{\omega \mapsto r_{ij}}[\omega] + \underbrace{\mathcal{I}(x(t_s),\bar{x}_{ij}(t_s))}_{\text{initial condition}},
\end{align}
where $\T^{\Sih}_{dy}$ and $\T^{\Sih}_{\omega y}$ are the LTI systems from the external signals~$(d,\omega)$ to $y$, and $\mathcal{I}(\cdot)$ is the contribution of the internal states of the system~$x(t_s)$ and the filter~$\bar{x}_{ij}(t_s)$. From the classical system theory, we know that the initial condition contribution vanishes exponentially fast under appropriate stability conditions. To isolate the active mode, we adopt the same mindset as GOS and opt to decouple the contribution of the reference signal~$d$ (i.e., the first term in the right-hand side of~\eqref{eq:Design form Fij}) for the matched mode $j = h$, and make sure that it is significantly high when $j \neq h$. With regards to the latter, we look at the steady-state behavior of the filters, motivated by the fact that in many important applications the reference signal~$d$
is effectively constant between two switching instants. Furthermore, we opt to suppress the noise contribution (the second term in the right-hand side of~\eqref{eq:Design form Fij}) for all $h \in \Nmodeset$. These steps will be formalized in the next part.

\begin{figure}[t]
    \centering
    \begin{minipage}{0.5\textwidth}
    \centering
    \includegraphics[scale=0.7]{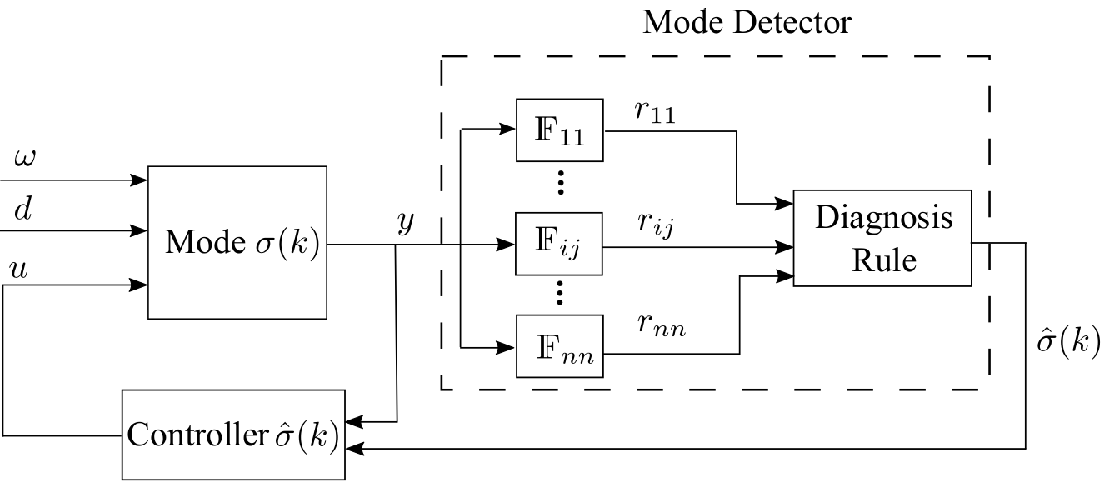} 
    \caption{Structure of the closed-loop dynamics and the mode detector}\label{fig:closed-loop with FDI}
    \end{minipage}
    \begin{minipage}{0.48\textwidth}
    \centering
    \includegraphics[scale=1.3]{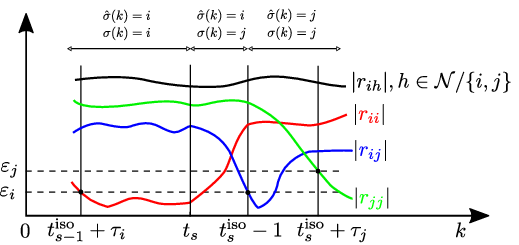} 
    \caption{Illustration of the diagnosis process}
    \label{fig:diagnosis process}
    \end{minipage}
\end{figure}

\subsection{Problem statements}
The proposed diagnosis solution comprises two components: (i)~bank of filters, as briefly described in the previous section, and (ii)~diagnosis rule, which is essentially a thresholding technique to estimate the system mode from the residuals. We then present two problems concerning each of these components. For each pair $(i,j)$ and the respective filter~$\Fij$, the desired properties of~$d \mapsto r_{ij}$ and~$\omega \mapsto r_{ij}$~(the first two terms in the right-hand side of~\eqref{eq:Design form Fij}) can be formalized as follows:
\begin{subequations} \label{eq:input-output relationships}
    \begin{align}
    \Fij \T^{\Sij}_{dy} &  = 0 ,  \label{eq:in-out relation 1} \\  
    \left| \left[\Fij \T^{\Sih}_{dy}\right]_{\rm ss} \right| & \geq 1, \qquad \forall h \in \Nmodeset \setminus \{j\}, \label{eq:in-out relation 2} \\
    \left\| \Fij \T^{\Sih}_{\omega y} \right\|_{\Ht} & \leq \eta_{ijh}, \quad \forall h \in \Nmodeset. \label{eq:in-out relation 3} 
    \end{align}
\end{subequations} 
Let us briefly elaborate on each condition in \eqref{eq:input-output relationships}: The equality constraint~\eqref{eq:in-out relation 1} decouples~$d$ from~$r_{ij}$ when the closed-loop status is $\Sij$. The condition~\eqref{eq:in-out relation 2} ensures that the absolute value of the steady-state gain of~$\Fij \T^{\Sih}_{dy}$ remains larger than or equal to~$1$, and as such, the contribution of~$d$ to~$r_{ij}$ is notably nonzero when the closed-loop status is~$\Sih$ for all~$h\in \Nmodeset\setminus\{j\}$. 
Recall that the~$\Ht$-norm of a transfer function is the asymptotic variance of the white noise response~\cite{scherer1997multiobjective}. Then, the constant~$\eta_{ijh}$ in \eqref{eq:in-out relation 3} is an upper bound for the variance of the noise contribution to the residual. In view of the desired properties~\eqref{eq:input-output relationships}, we proceed with our first problem. 

\noindent\textbf{Problem 1.} (Optimal bank of filters) Consider the closed-loop dynamics~\eqref{eq:system}-\eqref{eq:controller} and the mode detector in Figure~\ref{fig:closed-loop with FDI}. Given~$i,j \in \Nmodeset$ and an admissible family of the filters~$\Fij$, find the optimal filter defined through the optimization program 
\begin{align}\label{eq:problem1}
    \min_{\Fij,~\eta_{ijh}} \left\{\sum_{h=1}^{n} \eta_{ijh}:\eqref{eq:in-out relation 1},\eqref{eq:in-out relation 2},\eqref{eq:in-out relation 3} \right\}.
\end{align}

Given the filters as an (approximate) solution to~\eqref{eq:problem1}, we now shift our attention to the diagnosis rule component in Figure~\ref{fig:closed-loop with FDI}. Consider a transition from mode~$i$ to mode $j$ at time instant~$t_s$ (i.e.,~$\hsigma(t_s) = i, \sigma(t_s) = j$) where $i,j \in \Nmodeset$. There are two key parameters during the diagnosis process of the transition: (1)~the threshold~$\varepsilon_i \in \R_{+}$, and (2)~the waiting time~$\tau_j \in \N_{+}$. The behavior of the trajectories, as well as the design parameters~$\varepsilon_i, \tau_j$, are pictorially illustrated in Figure~\ref{fig:diagnosis process}. In the following, the role of each of the design parameters is discussed:

(1)~\textit{Threshold}~$\varepsilon_i$: 
As formalized in~\eqref{eq:input-output relationships}, the matched residual $r_{\hsigma \sigma}(k)$ should be close to zero, while the other residuals are notably away from zero. Recall that the current controller mode~$\hsigma(k)$ is known, and the system mode~$\sigma(k)$ is the detection target. Hence, we monitor the residuals~$r_{\hsigma h}(k)$ for all~$h \in \Nmodeset$, and compare them with the threshold $\varepsilon_{\hsigma}$ to isolate the matched residual (the one with the smallest absolute value). More specifically, we opt to single out one candidate from all the other possible modes. This procedure can be formally described by introducing the following conditions
\begin{subequations}\label{eq:diag rules}
    \begin{align}
        \label{eq:diag rules:j*}
        &j^*(k) = \arg\min_{h \in \Nmodeset} |r_{\hsigma h}(k)|, \\
        \label{eq:diag rules:eps}
        &|r_{\hsigma j^*}(k)| \leq  \varepsilon_{\hsigma(k)} < \min\limits_{h \in \Nmodeset \setminus \{j^*(k)\}}|r_{\hsigma h}(k)|.
    \end{align}
\end{subequations}
The mode~$j^*(k)$ defined in~\eqref{eq:diag rules:j*} is our best candidate to estimate the system mode~$\sigma(k)$, and \eqref{eq:diag rules:eps} is essentially a requirement to ensure that the threshold only selects one candidate. Once the conditions~\eqref{eq:diag rules} are fulfilled at a time instant~$k$, then the diagnosis component updates $\hsigma(k+1) = j^*(k)$, otherwise, it still retains the old mode $\hsigma(k+1) = \hsigma(k)$. In Figure~\ref{fig:diagnosis process}, note the period prior to $t^{\rm iso}_s$, the {\em isolation} time of the transition at~$t_s$; this will be formally defined in the next part in~\eqref{eq:isolation instant}.
 
(2)~\textit{Waiting time}~$\tau_j$: Once we update~$\hsigma$ at~$t^{\rm iso}_s$, the conditions~\eqref{eq:diag rules} are violated immediately since the controller mode changes. Thus, we need to wait for sufficiently large time to pass the transient behavior of the system caused by the initial condition (the third term in the right-hand side of~\eqref{eq:Design form Fij}); see the ``waiting period"~$[t^{\rm iso}_s, t^{\rm iso}_s+\tau_j)$ in Figure~\ref{fig:diagnosis process}. The controller mode~$\hsigma$ remains unchanged during this period~(i.e.,~$\hsigma(k+1) = \hsigma(k)$) until~$|r_{jj}(k)|$ reaches the respective threshold~$\varepsilon_j$; see Figure~\ref{fig:diagnosis process} and the time instant~$t^{\rm iso}_s + \tau_j$. To determine whether the diagnosis process is in the waiting period or not, we record the last isolation time instant through
\begin{align}\label{eq:isolation instant}
        t^{\rm iso}(k) := \max \Big\{t \in \N_{+}: \hsigma(t) \neq \hsigma(t-1),~k \geq t \Big\}.
\end{align}
We use the shorthand notation~$t^{\rm iso}(k) = t^{\rm iso}_s$ for~$k\in[t^{\rm iso}_s, t^{\rm iso}_{s+1})$.

In summary, the diagnosis rule of the second component can be mathematically described by
\begin{align}\label{eq:diagnosis rules}
\hsigma(k+1) &= \left\{ \begin{array}{l}
j^*(k), ~\text{if} 
     ~\eqref{eq:diag rules} ~\text{and} ~k \geq t^{\rm iso}(k)+\tau_{\hsigma (k)} \\
\hsigma (k), ~\text{otherwise}.
\end{array} \right. 
\end{align}

Note that~$\varepsilon_i$ in \eqref{eq:diag rules} and~$\tau_j$ in \eqref{eq:diagnosis rules} are the design parameters, and their objective is to detect the current system mode~$\sigma(k)$. In view of the update rule~\eqref{eq:diagnosis rules}, this objective is formalized in our next problem in terms of the behavior of the matched filter residual~$r_{ij}(k)$. 

\noindent\textbf{Problem 2.} (Probabilistic performance certificates) Suppose that the transition from mode $i$ to $j$ occurs at~$t_s$ (i.e.,~$\hsigma(t_s) = i$ and $\sigma(t_s) = j$). Given the filters constructed from Problem~1 and a reliability level~$\beta \in (0,1]$, determine the threshold~$\varepsilon_i$ and the estimated matched time~$T_{ij}$ such that
\begin{align}\label{eq:Prob guarantee}
   \mPr \left[\left|r_{ij}(t)\right| \leq \varepsilon_i  ~\Big|
   ~\begin{bmatrix} \hsigma(k)  \\ \sigma(k) \end{bmatrix}
   =\begin{bmatrix} i  \\ j \end{bmatrix},  k \geq t_s  \right] \geq 1-\beta, \quad\forall~t \geq t_s+T_{ij}.
\end{align}

The initial condition~$\mathcal{I}(x(t_s),\bar{x}_{ij}(t_s))$ determines the time that~$|r_{ij}|$ takes to reach~$\varepsilon_i$. However, the internal system state~$x(t_s)$ and switching instant~$t_s$ are unknown. Moreover, if the next transition occurs before $t_s+T_{ij}$, the guarantee in~\eqref{eq:Prob guarantee} is no longer useful. Thus, we assume that the time between two consecutive transitions (the so-called \textit{dwell time}~\cite{yuan2018novel}) is large enough so that the system reaches its steady-state before the next transition. It is a reasonable assumption as the dwell time of many real-world applications is longer than the time available for the controller to detect the mode. In this setting, the probabilistic guarantee~\eqref{eq:Prob guarantee} can be obtained, and the internal state~$x(t_s)$ can be estimated by its steady-state value.

\begin{remark}(Waiting time)
The waiting time~$\tau_j$ depicted in Figure~\ref{fig:diagnosis process} is indeed a special case of the estimated matched time introduced in Problem~2 where the controller and the system mode coincide, i.e., $\hsigma(t^{\rm iso}_s) = \sigma(t^{\rm iso}_s) = j$, and as such $\tau_j = T_{jj}$.
\end{remark}

\section{Main Result} \label{sec:main results}
In this section, the structure and design method of the filters are presented. Then, computation methods of the thresholds and the estimated matched time are given to provide probabilistic guarantees on the diagnosis performance.
All proofs are moved to Section 4 to improve readability.

\subsection{Filter design: optimization-based method}
Suppose the current status is~$\Sih$, i.e.,~$\hsigma(k)=i,\sigma(k)=h$. The closed-loop dynamics~\eqref{eq:system}-\eqref{eq:controller} can be written as
\begin{align}\label{eq:closed-loop dynamics}
	x(k+1) &= A^{cl}_{ih}x(k) + E_{h} d(k) +(W_h+B_{h}K_{i}D_h)\omega(k)  \notag\\
	y(k) &= C_{h} x(k) +D_h\omega(k),
\end{align}
where~$A^{cl}_{ih} = A_{h}+B_{h}K_{i}C_{h}$. 
We further reformulate \eqref{eq:closed-loop dynamics} into the DAE format, which is
\begin{equation}\label{eq:DAE model}
    H_{ih}(\mfq) \begin{bmatrix} x \\d \end{bmatrix}
    + L(\mfq) [y] + G_{ih}(\mfq)[\omega] = 0.
\end{equation}
The operator~$\mfq$ is a time-shift operator, i.e.,~$x(k+1)=\mfq x(k)$. The polynomial matrices~$H_{ih}(\mfq)$,~$L(\mfq)$ and~$G_{ih}(\mfq)$ are given by
\begin{align*}
&H_{ih}(\mfq) =  H_{ih,1} \mfq +H_{ih,0} 
=\begin{bmatrix}
	-\mfq I+A^{cl}_{ih} &E_h\\
	C_h              &0
\end{bmatrix}, \\
	&L(\mfq) = L_0 =  \begin{bmatrix} 0 \\ -I \end{bmatrix},
	~G_{ih}(\mfq) = G_{ih,0} = \begin{bmatrix}
	W_h+B_h K_i D_h \\ D_h
	\end{bmatrix}.
\end{align*}

Inspired by~\cite{nyberg2006residual} and~\cite{esfahani2015tractable}, the filter~$\Fij$ is defined as
\begin{equation}\label{eq:filter}
\Fij(\mfq) = a^{-1}(\mfq)N_{ij}(\mfq)L(\mfq),
\end{equation}
where the polynomial row vector~$N_{ij}(\mfq)=\sum^{d_N}_{m=0} N_{ij,m} \mfq^m$, each $N_{ij,m} \in \R^{1 \times (n_x +n_y)}$ is a constant row vector, $d_N$ denotes the degree of~$N_{ij}(\mfq)$, and~$a(\mfq)$ is a~$(d_N+1)$-th order polynomial with all roots inside the unit disk. We define
\begin{equation}\label{eq:a(p)}
a(\mfq)=\mfq^{d_N+1} + a_{d_N}\mfq^{d_N} + \dots + a_{1}\mfq + a_{0},
\end{equation}
where~$a_m$ is a constant coefficient for each $m \in \{0,1,\dots,d_N\}$. Notice that the role of~$a(\mfq)$ is to ensure that the filter~$\Fij$ is strictly proper and stable. 
To simplify the design process, we fix~$a(\mfq)$ and~$d_N$, and suppose that all the filters are of the same degree. The coefficients of the numerator, i.e.,~$N_{ij,m}$ for~$m \in \{0,1,\dots,d_N\}$, are the design parameters. Multiplying the left-hand side of~\eqref{eq:DAE model} by~$a^{-1}(\mfq)N_{ij}(\mfq)$ yields
\begin{equation}\label{eq:r_ij}
\begin{split}
r_{ij} = \frac{N_{ij}(\mfq)L(\mfq)}{a(\mfq)} [y] = -\frac{N_{ij}(\mfq)H_{ih}(\mfq)}{a(\mfq)} \begin{bmatrix} x \\d \end{bmatrix}  -\frac{N_{ij}(\mfq)G_{ih}(\mfq)}{a(\mfq)} [\omega].
\end{split}
\end{equation}
To bound the~$\Ht$-norm of the transfer function from~$\omega$ to~$r_{ij}$, we derive the observable canonical form of~$\Fij(\mfq)$ from~\eqref{eq:filter} 
\begin{align}\label{eq:Fij state-space model}
    \bar{x}_{ij}(k+1) &= A_r \bar{x}_{ij}(k) + B_{r_{ij}} y(k) \notag\\
     r_{ij}(k) &= C_r \bar{x}_{ij}(k),
\end{align}
where$~\bar{x}_{ij}(k) \in \R^{d_N+1}$ denotes the state, matrices~$A_{r},B_{r_{ij}},C_r$ are
\begin{align}\label{eq:ob realize}
&A_r = \begin{bmatrix}
0  &\dots &0 &-a_{0} \\
1  &\dots &0 &-a_{1} \\
\vdots  &\ddots &\vdots &\vdots \\
0  &\dots &1 &-a_{d_N} \\
\end{bmatrix}, 
~B_{r_{ij}} =  \begin{bmatrix}
N_{ij,0} \\ N_{ij,1} \\ \vdots \\ N_{ij,d_N}
\end{bmatrix} L_0, 
~C_r = \begin{bmatrix}
0 \dots 0~1
\end{bmatrix}.
\end{align}
The parameters~$N_{ij,m}$ are reformulated into $B_{r_{ij}}$ here. Let us introduce an augmented state $\X_{ij}(k) := \left[x(k)^{\top} ~\bar{x}_{ij}(k)^{\top}\right]^{\top}$. The dynamics of~$\X_{ij}$ can be derived from~\eqref{eq:closed-loop dynamics} and~\eqref{eq:Fij state-space model}, which is 
\begin{equation}\label{eq:closed-loop and filter}
\begin{split}
    \X_{ij}(k+1) &=  \mathcal{A}_{ijh} \X_{ij}(k) + \mathcal{E}_{h} d(k) + \mathcal{D}_{ijh} \omega(k)\\
    r_{ij}(k) &= \mathcal{C}_r \X_{ij}(k),
\end{split} 
\end{equation}
where
\begin{align*}
    &\mathcal{A}_{ijh} = \begin{bmatrix}A^{cl}_{ih} &0 \\B_{r_{ij}} C_h &A_r
    \end{bmatrix}, 
    ~\mathcal{E}_{h} = \begin{bmatrix}
    E_h \\ 0
    \end{bmatrix},
    ~\mathcal{D}_{ijh} = \begin{bmatrix}
    W_h+B_{h}K_{i}D_h \\ B_{r_{ij}}D_h
    \end{bmatrix},
    ~\mathcal{C}_r = \begin{bmatrix}
    0 ~C_r
    \end{bmatrix}.
\end{align*}

To design filters satisfying conditions in Problem 1, we formulate an optimization problem in the following theorem. For clarity of exposition, we allocate the first two lines to the decision variables in the optimization problem.
 
\begin{theorem}[Optimal bank of filters: exact finite reformulation]\label{thm:minimization problem}
Consider the closed-loop dynamics~\eqref{eq:system}-\eqref{eq:controller} and the filter~$\Fij$ proposed in~\eqref{eq:filter} with the state-space realization~$(A_r, B_{r_{ij}}, C_r)$ as defined in~\eqref{eq:ob realize}. Given the order~$d_N$, coefficients of~$a(\mfq)$, and a sufficiently small~$\vartheta \in \mathbb{R}_+$,
Problem 1 as defined in~\eqref{eq:problem1} can be equivalently translated into the following finite optimization program: 
\begin{subequations}\label{eq:minimization problem}
    \begin{align}
    &\min ~\sum_{h=1}^n \eta_{ijh} \notag \\
    \textup{s.t.} 
    ~&N_{ij,m} \in \R^{1 \times (n_x +n_y)}, \forall m \in \{0,1,\dots,d_N\},
    ~\eta_{ijh} \in \mathbb{R}_+, \forall h\in\Nmodeset, \notag\\
    ~&P_{ij}\in \mathcal{S}^{d_N+1}, 
    ~P_{ijh}\in \mathcal{S}^{n_x+d_N+1},\forall h\in\Nmodeset\setminus\{j\}  \notag \\
    ~&\bar{N}_{ij} \bar{H}_{ij} = 0,  \label{eq:constraint 1} \\ 
    ~& \left| a^{-1}(1) \bar{N}_{ij} \mathcal{L}_{ih} \right|\geq 1, ~\forall h\in\Nmodeset\setminus\{j\},  \label{eq:constraint 2}\\
    &\begin{bmatrix}
       P_{ij}  &A_r P_{ij} &\mathcal{B}_{ij} \\
        *      &P_{ij}     &0 \\
        *      &*          &I 
    \end{bmatrix} \succeq \vartheta I, 
    ~\begin{bmatrix}
        \eta_{ijj}    &C_r P_{ij} \\
        *             &P_{ij}
    \end{bmatrix} \succeq \vartheta I, \label{eq:constraint 3} \\
    & \begin{bmatrix}
        P_{ijh} &\mathcal{A}_{ijh}P_{ijh} &\mathcal{D}_{ijh}   \\
        *  &P_{ijh} &0 \\ 
        * &* &I    \end{bmatrix}  \succeq \vartheta I, 
    ~\begin{bmatrix}
    \eta_{ijh}    &\mathcal{C}_r P_{ijh} \\
    *             &P_{ijh}
    \end{bmatrix} \succeq \vartheta I, 
    ~\forall h\in\Nmodeset\setminus\{j\}. \label{eq:constraint 4}
\end{align}
\end{subequations}
where the involved matrices are given by
\begin{align*}
    &\bar{N}_{ij}= [N_{ij,0} ~N_{ij,1} ~\dots ~N_{ij,{d_N}}], ~\mathcal{L}_{ih} = \bar{L} [\overbrace{I ~\dots ~I}^{d_N+1}]^{\top} C_h\left(I-A^{cl}_{ih}\right)^{-1}E_h, \\
    &\bar{H}_{ij}= \begin{bmatrix} 
                    H_{ij,0} &H_{ij,1}  &\dots    &0 \\
	                \vdots   &\ddots    &\ddots   &\vdots\\
	                0        &\dots     &H_{ij,0} &H_{ij,1}
                    \end{bmatrix},
   ~\mathcal{B}_{ij} =  -\begin{bmatrix}
                     N_{ij,0}  \\ \vdots \\ N_{ij,d_N}
                     \end{bmatrix} G_{ij,0}, 
   ~\bar{L}=\begin{bmatrix}
                    L_0, &\dots, &0 \\ \vdots &\ddots &\vdots\\ 0 &\dots &L_0
             \end{bmatrix}.
\end{align*}
\end{theorem}
\begin{proof}
The proof is provided in Section 4.1.
\end{proof}

Note that if~$N^*_{ij,0},\dots,N^*_{ij,d_N}$ are feasible solutions to~\eqref{eq:minimization problem}, then so are~$-N^*_{ij,0},\dots,-N^*_{ij,d_N}$ with the same objective values. This directly holds for constraints~\eqref{eq:constraint 1} and~\eqref{eq:constraint 2} and can be proved through Schur complement for the matrix inequalities constraints~\eqref{eq:constraint 3} and~\eqref{eq:constraint 4}. Thus, we can drop the absolute value of~\eqref{eq:constraint 2} without loss of generality.


The following proposition shows that the nonlinear matrix inequality in~\eqref{eq:constraint 4} can be safely approximated with a LMI.  
\begin{proposition}[Optimal bank of filters: safe convex approximation]\label{prop:LMI Appro}
Consider the optimization problem~\eqref{eq:minimization problem}. Given~$\alpha\in\mathbb{R}$ and~$\gamma \in \mathbb{R}_+$, the nonlinear inequality constraint as the first term in~\eqref{eq:constraint 4} can be safely approximated by the following LMI constraint if there exist matrices~$\mathcal{G}_{ijh,1} \in \mathcal{M}^{n_x+d_N+1}$,~$\mathcal{G}_{ijh,2} \in \mathcal{M}^{n_{\omega}}$ such that:
\begin{align}\label{eq:LMI Appro}
    \begin{bmatrix}
    P_{ijh} &\hat{A}_{ih}\mathcal{G}_{ijh} &\hat{B}_{r_{ij}} &0 \\
    *  &\Xi_{ijh} &0 &\left(\hat{D}_h \mathcal{G}_{ijh}  \right)^{\top} \\
    *  &* &\frac{1}{\gamma}I &0 \\
    *  &* &* &\gamma I
\end{bmatrix} \succeq \vartheta I,
\end{align}
where the involved matrices are defined as
\begin{align*}
    &\hat{A}_{ih} = \begin{bmatrix} 
    \begin{bmatrix} A^{cl}_{ih} &0 \\ 0 &A_r \end{bmatrix} 
    &\begin{bmatrix} W_h+B_h K_i D_h \\ 0 \end{bmatrix} 
    \end{bmatrix}, 
    ~\mathcal{G}_{ijh} = \begin{bmatrix}
    \mathcal{G}_{ijh,1} &0 \\ 0 &\mathcal{G}_{ijh,2}
    \end{bmatrix}, \\
    &~\hat{D}_h = \begin{bmatrix} 
    \begin{bmatrix} C_h &0 \end{bmatrix} &D_h
    \end{bmatrix},~\hat{B}_{r_{ij}} = \begin{bmatrix}
    0 \\ -B_{r_{ij}}
    \end{bmatrix}, 
    ~\Xi_{ijh} = \alpha \mathcal{G}_{ijh} + \alpha \mathcal{G}_{ijh}^{\top}-\alpha^2\begin{bmatrix} P_{ijh} &0 \\ * &I \end{bmatrix}.
\end{align*}
\end{proposition}
\begin{proof}
The proof is provided in Section 4.1.
\end{proof}

It is worth pointing out that the linear approximation~\eqref{eq:LMI Appro} provides a sufficient condition for the nonlinear matrix inequality in~\eqref{eq:constraint 4}. This means that any feasible solution to~\eqref{eq:LMI Appro} is necessarily a feasible solution to the nonlinear constraint.

Furthermore, we provide necessary and sufficient conditions for the feasibility of the optimization problem~\eqref{eq:minimization problem} in the following proposition. Here, the rank and eigenvalues of a matrix~$A$ are denoted by~$\text{Rank}(A)$ and~$\Lambda(A)$, resp.
\begin{proposition}[Optimal bank of filters: feasibility]\label{prop:feasibility}
The optimization problem~\eqref{eq:minimization problem} is feasible if and only if the following conditions are satisfied.
\begin{subequations}
    \begin{align}
    &(d_N+1)(n_x+n_y) > \textup{Rank}\left(\bar{H}_{ij}\right), \label{eq:lower bound dN}\\
    &\textup{Rank}\left(\left[\bar{H}_{ij} ~\mathcal{L}_{ih}\right]\right) > \textup{Rank}\left(\bar{H}_{ij}\right), \label{eq:rank condition}\\
    & |\Lambda(A_r)| < 1,~|\Lambda(A^{cl}_{ih})|<1, ~\text{\normalfont (i.e.,~$A_r$ and~$A^{cl}_{ih}$ are stable)}. \label{eq:eig-of-A}
    \end{align}
\end{subequations}
\end{proposition}
\begin{proof}
The proof is provided in Section 4.1.
\end{proof}
Note that the inequality~\eqref{eq:lower bound dN} provides a way to find the minimum filter order~$d_N$.
According to~\eqref{eq:eig-of-A},~$|\Lambda(A^{cl}_{ih})|<1$ ensures that~\eqref{eq:constraint 4} is feasible. However,~$A^{cl}_{ih}$ could be unstable because the model and controller are unmatched. Hence, the constraints in~\eqref{eq:constraint 4} with unstable~$A^{cl}_{ih}$ should be excluded. Since the unmatched residuals of those unstable modes diverge from zero, removing those constraints does not affect the mode detection task. 

\begin{remark}[Observability]
The conditions~\eqref{eq:lower bound dN} and~\eqref{eq:rank condition} are related to the observability of switched affine systems theoretically~\cite{kusters2018switch}. In particular, the mode can be determined deterministically without noise if the two conditions are satisfied.  
Also, observability of each mode is not necessary, which is consistent with the result in~\cite[Theorem~8]{kusters2018switch}.
\end{remark}

We close this section with the following remark on different sources of conservatism for the proposed filter design: 
\begin{remark}[Conservatism analysis]
The conservatism of the proposed approximate solution stems from three different sources: 
\begin{enumerate}[label=(\roman*), itemsep = 0mm, topsep = 0mm, leftmargin = 4mm]

\item Reference signal dimension: We only focus on one-dimensional reference signals, but instead, we do not require any further prior assumptions on their values. As shown in~\cite{pan2019static}, this restriction is inevitable when the filter residual is one-dimensional since different elements of a multi-dimensional reference signal may cancel out each other's contributions. 
 
\item Filters denominator: To simplify the design, the filters denominator~$a(\mfq)$ are all fixed, which reduces design freedom.

\item Non-convexity: The exact reformulation of Problem~1 is a non-convex optimization problem (Theorem~\ref{thm:minimization problem}), for which we propose a safe convex approximation (Proposition~\ref{prop:LMI Appro}).
\end{enumerate}
\end{remark}


\subsection{Performance certificates}
With the filters designed by using~\eqref{eq:minimization problem}, we now determine the threshold~$\varepsilon_i$ and waiting time~$\tau_j$ to ensure proper detection task governed by~\eqref{eq:diagnosis rules}. Considering the stochastic measurement noise~$\omega$, we resort to the probabilistic guarantees depicted in~\eqref{eq:Prob guarantee}. Let us introduce the following lemma and assumption. 

\begin{lemma}[Sub-Gaussian concentration {\cite[Proposition 2.5.2]{vershynin2018high}}] \label{lem:subG}
Suppose~$\chi$ is an $\R^{n_{\chi}}$-valued sub-Gaussian random vector with positive parameter~$\zeta$, i.e., $\mE\left[\rme^{\phi \nu^{\top}(\chi-\mE[\chi])}\right] \leq {\rm e}^{\zeta^2 \phi^2 /2}$ for all $\phi \in \R$ and $\nu \in \R^{n_{\chi}}$ where $\|\nu\|_2 =1$. Then, we have
\begin{align}\label{eq:Concentra Ine}
    \mPr\Big[\|\chi - \mE[\chi]\|_\infty \leq \varepsilon \Big] \geq 1- 2 n_{\chi}\, {\rm e}^{-\varepsilon^2/({2{\zeta}^2})}, \quad \forall \varepsilon \in \mathbb{R}_+.
\end{align}
\end{lemma}

\begin{assumption}[Sub-Gaussian noise]\label{as:subG}
The measurement noise~$\omega$ is an iid sub-Gaussian signal with zero mean and parameter~$\zeta_{\omega} \in \mathbb{R}_+$ as defined in Lemma~\ref{lem:subG}.
\end{assumption}
 
From~\eqref{eq:Concentra Ine}, the tails of sub-Gaussian distributions decay exponentially. Moreover, the class of sub-Gaussian distributions is broad, containing Gaussian, Bernoulli, and all bounded distributions. In the following results, the noise is assumed to be sub-Gaussian.
To improve readability, we further introduce several notations. Let the polynomial row vector~$N_{ij}(\mfq) := \left[\hat{N}_{ij}(\mfq) ~\check{N}_{ij}(\mfq)\right]$, where~$\hat{N}_{ij}(\mfq)$ and~$\check{N}_{ij}(\mfq)$ have dimensions~$n_x$ and~$n_y$, resp.
Define~$\lambda_{\max} := \max_{m\in \{1,\dots,d_{N+1}\}} |\lambda_{m}|$, where~$\lambda_m$ is a root of~$a(\mfq)$ defined in~\eqref{eq:a(p)}. 
These roots are chosen to be distinct, i.e.,~$\lambda_m \neq \lambda_n$ for~$m \neq n$. The following theorem provides conditions for the probabilistic performance certificates.

\begin{theorem}[Probabilistic performance certificates]\label{thm:Prob certificate}
Suppose Assumption~\ref{as:subG} holds and the dwell time is large enough. Consider the closed-loop dynamics~\eqref{eq:system}-\eqref{eq:controller} and the filter~$\Fij$ with the poles ${\lambda_m}, m\in\{1,\dots,d_{N+1}\}$, and the numerator designed by using~\eqref{eq:minimization problem} with the corresponding optimal solutions~$\eta^*_{ijj}$.
Given a reliability level~$\beta \in (0,1]$ and a constant~$\mu \in \mathbb{R}_+$, the probabilistic performance~\eqref{eq:Prob guarantee} in Problem 2 is satisfied, if the threshold~$\varepsilon_i$ is set as
\begin{align}\label{eq:threshold}
 \varepsilon_i = \left(\mu + \zeta_{\omega}\sqrt{2 \ln{(2/\beta)}}\right) \sqrt{\bar{\eta}_i}
 ,\quad\bar{\eta}_i=\max_{j\in\Nmodeset} \eta^*_{ijj},
\end{align}
and the estimated matched time~$T_{ij}$ equals to
\begin{align}\label{eq:T_ij}
    T_{ij} = \left\lceil \frac{\log \left( \psi_{ij}\left(\Fij,\X_{ij}(t_s) \right) / \left(\mu\sqrt{\bar{\eta}_{i}} \right) \right)}{\log \lambda^{-1}_{\max}}  \right\rceil,
\end{align}
where~{\small$\psi_{ij}\left(\Fij,\X_{ij}(t_s)\right) = \sqrt{d_N+1} \left(1 + \lambda_{\max}^{-1} \| \mathsf{B}_{ij} \|_2\right) \left\| \mE\left[\X_{ij}(t_s)\right] \right\|_{2}$}. The matrix~$\mathsf{B}_{ij}$ is defined as
\begin{align*}
\mathsf{B}_{ij} = \begin{bmatrix}
b_{ij,11}         &\dots  &b_{ij,1n_x} \\
\vdots            &\ddots &\vdots \\
b_{ij,(d_N+1) 1}  &\dots  &b_{ij,(d_N+1)n_x}
\end{bmatrix},
\end{align*}
where~$b_{ij,\ell h} = -\sum^{d_N}_{m=0}\hat{N}_{ij,m}(h) \lambda_{\ell}^{m} \big/ \left(\prod_{\tilde{\ell} \neq \ell} (\lambda_{\tilde{\ell}} - \lambda_\ell)\right)$ for~$h\in\{1,\dots,n_x\},~\ell, \tilde{\ell}\in \{1,\dots,d_N+1\} $, and~$\hat{N}_{ij,m}(h)$ denotes the~$h$-th element of~$\hat{N}_{ij,m}$.
\end{theorem}
\begin{proof}
The proof is provided in Section 4.2.
\end{proof}

\begin{figure}[t!]
\center
\includegraphics[scale=1.2]{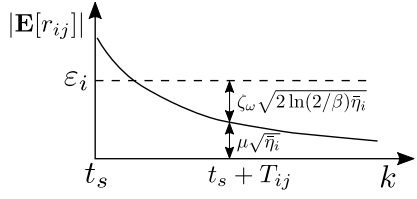}
\caption{Estimated matched time}\label{fig:T_ij}
\end{figure} 

The estimated matched time~$T_{ij}$ in~\eqref{eq:T_ij} is actually an upper bound for the time that~$|\mE[r_{ij}]|$ takes to arrive at~$\mu \sqrt{\bar{\eta}_i}$ after transition happens (as shown in Figure~\ref{fig:T_ij}). Then, we set the confidence interval according to $\beta$, such that~$\varepsilon_i$ is determined.

\begin{remark}[Threshold vs estimated matched time trade-off]
There is a trade-off in selecting~$\mu$ and~$\beta$ in~\eqref{eq:threshold}:
A smaller threshold~$\varepsilon_i$ provides high guarantees on excluding the unmatched residuals. We can decrease~$\varepsilon_i$ with smaller~$\mu$ or larger~$\beta$ from~\eqref{eq:threshold}. However, a smaller~$\mu$ can lead to a more conservative estimated matched time~$T_{ij}$ from~\eqref{eq:T_ij}. Also, a larger~$\beta$ increases the chance of false isolation. 
\end{remark}

\begin{remark}[Comparison with Chebyshev based bounds]
We highlight that the threshold~\eqref{eq:threshold} depends logarithmically on the reliability parameter, i.e.,~$\sqrt{\ln(2/\beta)}$. This is a significant improvement in comparison with the results based on Chebyshev's inequality (e.g.,~\cite[Section III.B]{boem2018plug}) in which the threshold scales polynomially by the factor~$\sqrt{1/\beta}$. 
\end{remark}

As a special case of~$T_{ij}$ in Theorem~\ref{thm:Prob certificate}, the waiting time~$\tau_j$ can be determined by 
\begin{align}\label{eq:tauj}   
   \tau_{j} = \left\lceil\frac{\log \left( \psi_{jj} \left(\Fjj,\X_{jj}(t^{\rm iso}_s)\right) / \left(\mu \sqrt{\bar{\eta}_j}\right)  \right)}{\log \lambda^{-1}_{\max} } 
   \right\rceil,
\end{align}
where {\small$\psi_{jj}\left(\Fjj,\X_{jj}(t^{\rm iso}_s)\right) = \sqrt{d_N+1} \left(1 + \lambda_{\max}^{-1} \| \mathsf{B}_{jj} \|_2\right) \left\| \mE\left[\X_{jj}\left(t^{\rm iso}_s\right)\right] \right\|_{2}$.} 

Observe that the expected values of~$\X_{ij}(t_s)$ and~$\X_{jj}(t^{\rm iso}_s)$ are required in~\eqref{eq:T_ij} and~\eqref{eq:tauj}. Recall that we assume that the dwell time is large enough and the system can reach the steady state before the next transition. The constant reference signal~$d$ is considered during the dwell time, i.e.,~$d(k) = \bar{d}$ for~$k\in[t_s,t_{s+1})$.
Then~$\mE\left[\X_{ij}(t_s)\right]$ can be estimated by its steady-state value~$\mE\left[\X_{ij}(t_s)\right] = \left(I-\mathcal{A}_{iji}\right)^{-1} \mathcal{E}_i \bar{d}$. For~$\mE\left[\X_{jj}(t^{\rm iso}_s)\right]$, since the actual diagnosis time is a random value, we first compute the steady-state value of~$\mE\left[\X_{jj}(t_s)\right]$. Then, according to the dynamics~\eqref{eq:closed-loop and filter}, we compute~$\max\limits_{i \in \Nmodeset} \left\| \mE\left[\X_{jj}(t_s+T_{ij})\right]\right\|_2$ as an estimation of~$\mE[\X_{jj}(t^{\rm iso}_s)]$.

According to the diagnosis rule~\eqref{eq:diagnosis rules}, one still needs to let the unmatched residuals be outside the threshold interval. Suppose the status is~$\Sij$. Inspired by the active fault diagnosis method~\cite{scott2014input}, we can design the reference signal~$d$ such that the unmatched residuals~$r_{ih}$ satisfy~$\left|\mE[r_{ih}]\right| \geq \varepsilon_i + \bar{\mu} \sqrt{\bar{\eta}_i}$ in the steady state, where~$\bar{\mu}\in\mathbb{R}_+$ is a constant. 
From the closed-loop dynamics~\eqref{eq:closed-loop dynamics} and~\eqref{eq:r_ij}, the expected value of~$r_{ih}$ can be written as
\begin{equation}\label{eq:r_ij_unmatched}
\mE[r_{ih}] = \frac{N_{ih}(\mfq)L(\mfq)}{a(\mfq)}C_j\left(\mfq I-A^{cl}_{ij}\right)^{-1} E_j \bar{d}.
\end{equation}
According to~\eqref{eq:r_ij_unmatched}, the requirement~$\left|\mE[r_{ih}]\right|  \geq \varepsilon_i + \bar{\mu} \sqrt{\bar{\eta}_i}$ is equivalent to choosing~$\bar{d}$ such that 
\begin{align}\label{eq:cond d}
\left|a^{-1}(1) \bar{N}_{ih} \mathcal{L}_{ij} \bar{d} \right| \geq   
 \varepsilon_i + \bar{\mu} \sqrt{\bar{\eta}_i}.
\end{align}
In the light of Lemma~\ref{lem:subG}, we have~$|r_{ih}|>\varepsilon_i$ with guaranteed probability in the steady state if~\eqref{eq:cond d} is satisfied.  

\begin{remark}[Regularities on the reference input]
When designing the filters and thresholds, the value of the reference signal is not necessary. 
However, this value is required when computing the estimated matched time~$T_{ij}$. Moreover, in order to separate the residuals of different modes in the presence of noise, the reference signal~$\bar{d}$ should satisfy~\eqref{eq:cond d}. Such constraint is not novel in the distinguishability problem for switched affine systems~\cite{rosa2011distinguishability}. This also can be interpreted as the persistence of excitation.
\end{remark}

\section{Technical Proofs of Main Results}\label{sec:proofs}
This section presents the technical proofs of the theoretical results in Section 3.

\subsection{Proofs of results in filter design}
Let us start with two lemmas required for the proof of Theorem~\ref{thm:minimization problem}.

\begin{lemma}(Multiplication of polynomial matrices \cite[Section III-A]{esfahani2015tractable})
\label{lem:poly_matrix_multiplication}
Let$~Q_1(\mfq)$ and$~Q_2(\mfq)$ be polynomial matrices of degree$~d_1$ and$~d_2$, resp., and defined by 
\begin{equation*}
	Q_1(\mfq) = \sum_{m=0}^{d_1} Q_{1,m}\mfq^{m}, \quad Q_2(\mfq) = \sum_{m=0}^{d_2} Q_{2,m}\mfq^{m},
\end{equation*} 
where $Q_{1,m} \in \R^{n_1 \times n_2}$ and $Q_{2,m} \in \R^{n_2 \times n_3}$ are the matrices of constant coefficients. The multiplication of$~Q_1(\mfq)$ and$~Q_2(\mfq)$ is equivalent to  
\begin{align*}
	Q_1(\mfq)Q_2(\mfq) = \bar{Q}_1\bar{Q}_2 \left[I ~\mfq I ~\dots ~\mfq^{d_1+d_2}I\right]^{\top},
\end{align*}
where~$\bar{Q}_1 = \begin{bmatrix} Q_{1,0} &Q_{1,1} &\dots   &Q_{1,{d_1}} \end{bmatrix}$ and
\begin{align*}
\bar{Q}_2 = \begin{bmatrix} Q_{2,0} &Q_{2,1} &\dots   &Q_{2,{d_2}} &0   &\dots &0 \\
 0   &Q_{2,0} &Q_{2,1} &\dots    &Q_{2,{d_2}}  &0  &\vdots \\
\vdots   &       &\ddots       &\ddots   &       &\ddots &0\\
 0   &0    &\dots   &Q_{2,0} &Q_{2,1}  &\dots  &Q_{2,{d_2}}
	\end{bmatrix}.
\end{align*}
\end{lemma}

The following lemma is a slight modification of the standard result concerning the~$\mathcal{H}_2$-norm of the stable LTI systems.

\begin{lemma}($\mathcal{H}_2$-norm~\cite[Lemma 1]{de2002extended})
\label{lem:H2_norm}
Consider the linear transfer function~$\T(\mfq) = C(\mfq I-A)^{-1}B$. For any constant~$\eta$, the system is stable and $\left\| \T(\mfq) \right\|^2_{\Ht} < \eta$ if and only if for all sufficiently small $\vartheta \in \R_+$, there exist~$P, Z \in \mathcal{S}$ such that the following LMIs are feasible:
\begin{equation*}
    \begin{bmatrix}
        P &AP &B \\
        * &P &0\\
        * &* &I
    \end{bmatrix} \succeq \vartheta I,
    ~\begin{bmatrix}
        Z &CP\\
        * &P
    \end{bmatrix} \succeq \vartheta I, 
    ~{\rm Trace}(Z) \leq \eta -\vartheta.
\end{equation*}
\end{lemma}

\begin{proof}[Proof of Theorem~\ref{thm:minimization problem}]
First, we show that the equality~\eqref{eq:constraint 1} guarantees the satisfaction of the property~\eqref{eq:in-out relation 1}. 
According to Lemma~\ref{lem:poly_matrix_multiplication}, it holds that
\begin{align*}
    N_{ij}(\mfq)H_{ij}(\mfq)=\bar{N}_{ij}\bar{H}_{ij}\left[I~\mfq I~\dots~\mfq^{d_N+1}I\right]^{\top}.
\end{align*}
Hence,~\eqref{eq:constraint 1} implies that~$N_{ij}(\mfq)H_{ij}(\mfq) = 0$.
The contribution of~$d$ to~$r_{ij}$ is completely canceled when the status is~$\Sij$ ($h=j$ in~\eqref{eq:r_ij}). This concludes the first part of the proof.

In the second part of the proof, we show that the constraint~\eqref{eq:constraint 2} implies the satisfaction of the property~\eqref{eq:in-out relation 2}. 
Suppose the status is~$\Sih$. According to the closed-loop dynamics~\eqref{eq:closed-loop dynamics}, we have
\begin{align*}
    y = C_h\left(\mfq I-A^{cl}_{ih}\right)^{-1} E_h [d] 
    + \left[C_h\left(\mfq I-A^{cl}_{ih}\right)^{-1} (W_h+B_hK_iD_h) + D_h\right] [\omega].
\end{align*}
By virtue of~\eqref{eq:r_ij} and the expression of~$y$, the transfer function from~$d$ to~$r_{ij}$ can be written as
\begin{align*}
\Fij\T^{\Sih}_{dy}(\mfq)
&= a^{-1}(\mfq)N_{ij}(\mfq)L(\mfq) C_h\left(\mfq I-A^{cl}_{ih}\right)^{-1}E_h \\
&=a^{-1}(\mfq)\bar{N}_{ij}\bar{L} \left[I ~\mfq I ~\dots ~\mfq^{d_N}I\right]^{\top}C_h \left(\mfq I-A^{cl}_{ih}\right)^{-1}E_h, 
\end{align*}
where Lemma~\ref{lem:poly_matrix_multiplication} is used in the second equality. Then, we enforce the absolute value of the steady-state gain of~$\Fij\T^{\Sih}_{dy}$ to be larger than or equal to 1 when~$h \neq j$, which is
\begin{align*}
   \left| \left[\Fij\T^{\Sih}_{dy}\right]_{ss} \right| = \left| a^{-1}(1) \bar{N}_{ij} \mathcal{L}_{ih} \right| \geq 1.
\end{align*}
This concludes the second part of the proof.
 
In the third part, we show that the inequalities \eqref{eq:constraint 3} and \eqref{eq:constraint 4} enforce the desired property~\eqref{eq:in-out relation 3}. 
When the status is~$\Sij$, as shown in~\eqref{eq:r_ij}, the transfer function from~$\omega$ to~$r_{ij}$  becomes
\begin{equation}\label{eq:T^Sij_wrij}
    \Fij\T^{\Sij}_{\omega y}(\mfq) =-a^{-1}(\mfq) N_{ij}(\mfq)G_{ij}(\mfq),
\end{equation}
where~$[x^{\top} ~d^{\top}]^{\top}$ is decoupled by~\eqref{eq:in-out relation 1}. Let~$(A_r,\mathcal{B}_{ij},C_r)$ be the observable canonical realization of~\eqref{eq:T^Sij_wrij}, whose derivation process is similar to that of~\eqref{eq:ob realize}. According to Lemma~\ref{lem:H2_norm}, the inequalities~\eqref{eq:constraint 3} imply~$\left\| \Fij\T^{\Sij}_{\omega y} \right\|^2_{\Ht} < \eta_{ijj}$ directly. Note that the slack variable $Z$ shown in Lemma~\ref{lem:H2_norm} has one dimension in this problem, thus the third inequality is dropped. 
When the status is~$\Sih$ for~$h\in \Nmodeset\setminus\{j\}$, the transfer function from~$\omega$ to~$r_{ij}$ can be obtained from~\eqref{eq:closed-loop and filter}. Again, according to Lemma~\ref{lem:H2_norm}, the inequalities~\eqref{eq:constraint 4} imply~$\left\| \Fij\T^{\Sih}_{\omega y} \right\|^2_{\Ht} < \eta_{ijh}$. 
Then we take the sum of~$\eta_{ijh}$ for all~$h\in\Nmodeset$ as the objective function to minimize the effect of~$\omega$ on~$r_{ij}$.
This completes the proof. 
\end{proof}

\begin{proof}[Proof of Proposition~\ref{prop:LMI Appro}]
This proof is to show that~\eqref{eq:LMI Appro} ensures the satisfaction of the nonlinear matrix inequality in~\eqref{eq:constraint 4}. By applying Schur complement to~\eqref{eq:LMI Appro}, we have
\begin{align}\label{eq:LMI ineq-1}
     &\begin{bmatrix}
     P_{ijh} &\hat{A}_{ih}\mathcal{G}_{ijh}  \\
     *       &\Xi_{ijh} \end{bmatrix} 
    -\begin{bmatrix}
     \hat{B}_{r_{ij}}  &0 \\  
     *                 &\left(\hat{D}_h\mathcal{G}_{ijh}  \right)^{\top} \end{bmatrix}
    \begin{bmatrix} \gamma I &0 \\ 
                    *        &\frac{1}{\gamma} I  
    \end{bmatrix}
    \begin{bmatrix} \hat{B}_{r_{ij}}^{\top} &0 \\ * &\hat{D}_h\mathcal{G}_{ijh}  \end{bmatrix}  \notag\\
    = &\begin{bmatrix}
     P_{ijh} &\hat{A}_{ih}\mathcal{G}_{ijh}  \\
    *  &\Xi_{ijh} \end{bmatrix} 
    -\gamma \begin{bmatrix} \hat{B}_{r_{ij}} \\ 0  \end{bmatrix} 
    \begin{bmatrix} \hat{B}_{r_{ij}}^{\top} ~0  \end{bmatrix} 
    - \frac{1}{\gamma}  \begin{bmatrix} 0  \\ \left(\hat{D}_h\mathcal{G}_{ijh} \right)^{\top}  \end{bmatrix} 
    \begin{bmatrix} 0 ~\hat{D}_h \mathcal{G}_{ijh}  \end{bmatrix}
    \succeq  \vartheta I. 
\end{align}
Note that, for matrices~$A, B$ with appropriate dimensions and any scalar~$\gamma > 0$, it holds that~$ \gamma A A^{\top}+ \frac{1}{\gamma}B^{\top}B \succeq AB+B^{\top}A^{\top} $~\cite[Lemma 1]{chang2013new}. We have
\begin{align*}
    &-\begin{bmatrix} \hat{B}_{r_{ij}} \\ 0  \end{bmatrix} 
      \begin{bmatrix} 0 &\hat{D}_h \mathcal{G}_{ijh}  \end{bmatrix}
   -\begin{bmatrix} 0  \\ \left(\hat{D}_h \mathcal{G}_{ijh}\right)^{\top}  \end{bmatrix}    \begin{bmatrix} \hat{B}_{r_{ij}}^{\top} &0  \end{bmatrix} \\
   \succeq
    &-\gamma \begin{bmatrix} \hat{B}_{r_{ij}} \\ 0  \end{bmatrix} 
    \begin{bmatrix} \hat{B}_{r_{ij}}^{\top} &0  \end{bmatrix} - \frac{1}{\gamma}  \begin{bmatrix} 0  \\ \left(\hat{D}_h \mathcal{G}_{ijh} \right)^{\top}  \end{bmatrix} 
    \begin{bmatrix} 0 &\hat{D}_h \mathcal{G}_{ijh}  \end{bmatrix} .
\end{align*}
Thus, the inequality~\eqref{eq:LMI ineq-1} can be written as
\begin{align}\label{eq:LMI ineq-3}
    &\begin{bmatrix}
    P_{ijh} &\hat{A}_{ih}\mathcal{G}_{ijh}  \\
    *  &\Xi_{ijh} \end{bmatrix} 
    -\begin{bmatrix} \hat{B}_{r_{ij}} \\ 0  \end{bmatrix} 
      \begin{bmatrix} 0 ~\hat{D}_h\mathcal{G}_{ijh}  \end{bmatrix}
   -\begin{bmatrix} 0  \\ \left(\hat{D}_h\mathcal{G}_{ijh}\right)^{\top}  \end{bmatrix}    \begin{bmatrix} \hat{B}_{r_{ij}}^{\top} ~0\end{bmatrix} \notag\\
   =&\begin{bmatrix}
     P_{ijh} &\hat{A}_{ih}\mathcal{G}_{ijh}-\hat{B}_{r_{ij}}\hat{D}_h\mathcal{G}_{ijh}  \\
    *  &\Xi_{ijh} \end{bmatrix} \succeq  \vartheta I.
\end{align}
Expanding~$\hat{A}_{ih}\mathcal{G}_{ijh} - \hat{B}_{r_{ij}}\hat{D}_h\mathcal{G}_{ijh}$ leads to
\begin{align}\label{eq:LMI ineq-4}
&\begin{bmatrix} \begin{bmatrix} A^{cl}_{ih} &0 \\ 0 &A_r \end{bmatrix} \mathcal{G}_{ijh,1}
                 &\begin{bmatrix} W_h+B_h K_i D_h \\ 0 \end{bmatrix} \mathcal{G}_{ijh,2}  \end{bmatrix} 
- \begin{bmatrix} \begin{bmatrix} 0 \\ -B_{r_{ij}} \end{bmatrix} \begin{bmatrix} C_h &0 \end{bmatrix} \mathcal{G}_{ijh,1} & \begin{bmatrix} 0 \\ -B_{r_{ij}} \end{bmatrix}D_h \mathcal{G}_{ijh,2}  \end{bmatrix} \notag\\
= &\begin{bmatrix}
 \begin{bmatrix}   A^{cl}_{ih} &0 \\ B_{r_{ij}}C_h &Ar  \end{bmatrix}\mathcal{G}_{ijh,1}
 &\begin{bmatrix}   W_h+B_h K_i D_h \\ B_{r_{ij}}D_h  \end{bmatrix} \mathcal{G}_{ijh,2}
\end{bmatrix} 
= \begin{bmatrix}
 \mathcal{A}_{ijh} &\mathcal{D}_{ijh} \end{bmatrix} \mathcal{G}_{ijh}.
\end{align}
From~\eqref{eq:LMI ineq-4},~the inequality~\eqref{eq:LMI ineq-3} is equivalent to
\begin{align}\label{eq:LMI ineq-5}
    \begin{bmatrix}
     P_{ijh} &\begin{bmatrix}  \mathcal{A}_{ijh} &\mathcal{D}_{ijh} \end{bmatrix} \mathcal{G}_{ijh} \\
    *  &\Xi_{ijh} \end{bmatrix} \succeq  \vartheta I.
\end{align}
For a scalar~$\alpha \in \mathbb{R}$, matrices~$A,B$ with appropriate dimensions, and~$A\succ 0$, note that~$(B-\alpha A)^{\top}A^{-1}(B-\alpha A) \succeq 0$ implies~$B^{\top}A^{-1}B \succeq  \alpha B +\alpha B^{\top} -\alpha^2 A$. Thus, we have
\begin{align}\label{eq:LMI ineq-6}
\begin{split}
    {\mathcal{G}_{ijh}}^{\top} \begin{bmatrix} P_{ijh} &0 \\ * &I \end{bmatrix}^{-1} \mathcal{G}_{ijh} \succeq \Xi_{ijh}.    
\end{split}
\end{align}
By combining~\eqref{eq:LMI ineq-5} and~\eqref{eq:LMI ineq-6}, we obtain
\begin{align}\label{eq:LMI ineq-7}
    \begin{bmatrix}
     P_{ijh} &\begin{bmatrix}  \mathcal{A}_{ijh} &\mathcal{D}_{ijh} \end{bmatrix} \mathcal{G}_{ijh}  \\
    *  &{\mathcal{G}_{ijh}}^{\top} \begin{bmatrix} P_{ijh} &0 \\ * &I \end{bmatrix}^{-1} \mathcal{G}_{ijh} \end{bmatrix} \succeq  \vartheta I.
\end{align}
Pre- and post-multiplying~\eqref{eq:LMI ineq-7} by $\text{diag}(I,{\mathcal{G}_{ijh}}^{-\top})$ and $\text{diag}(I,P_{ijh},I)$ and their transpose successively, we arrive at
\begin{align*}
    \begin{bmatrix}
     P_{ijh} &\mathcal{A}_{ijh} P_{ijh} &\mathcal{D}_{ijh}   \\
    *  &P_{ijh} &0 \\ 
    * &* &I    \end{bmatrix}  \succeq \vartheta I.
\end{align*}
This completes the proof.
\end{proof} 

\begin{proof}[Proof of Proposition~\ref{prop:feasibility}]
We first show that the inequality~\eqref{eq:lower bound dN} is a necessary and sufficient condition for the constraint~\eqref{eq:constraint 1} having non-trivial solutions. 
According to Rank Plus Nullity Theorem~\cite[Chapter 4]{meyer2000matrix}, it holds that~$(d_N+1)(n_x+n_y)  = \textup{Rank}\left(\bar{H}_{ij}\right) + \textup{Null}\left(\bar{H}_{ij}\right)$, where~$\textup{Null}\left(\bar{H}_{ij}\right)$ denotes the dimension of the left null space of~$\bar{H}_{ij}$.
Thus, ~\eqref{eq:constraint 1} having non-trivial solutions is equivalent to~$\textup{Null}\left(\bar{H}_{ij}\right)$ being nonzero. 
This concludes the first part of the proof.

Second, we show that~\eqref{eq:rank condition} is equivalent to~\eqref{eq:constraint 2} when~\eqref{eq:lower bound dN} holds.  
~$(\Rightarrow)$ We proceed with the proof by contradiction. Suppose that~\eqref{eq:constraint 2} holds but~\eqref{eq:rank condition} is not satisfied, we have~$\textup{Rank}\left(\left[\bar{H}_{ij} ~\mathcal{L}_{ih}\right]\right) = \textup{Rank}\left(\bar{H}_{ij}\right)$. This means that~$\mathcal{L}_{ih}$ belongs to the column range space of~$\bar{H}_{ij}$. In other words, there exists a vector~$\xi \in \R^{(n_x+n_d)(d_N+2)}$, such that~$\mathcal{L}_{ih} = \bar{H}_{ij}\xi$. Since~$\bar{N}_{ij}$ satisfying~$\bar{N}_{ij}\bar{H}_{ij} = 0$, we have~$\bar{N}_{ij} \mathcal{L}_{ih} = \bar{N}_{ij}\bar{H}_{ij} \xi = 0$, which contradicts to~\eqref{eq:constraint 2}. ~$(\Leftarrow)$ Assume that~\eqref{eq:rank condition} holds. This means that the left null space of~$\bar{H}_{ij}$ and~$\mathcal{L}_{ih}$ are not the same. Thus, one can find a~$\bar{N}_{ij}$ which satisfies~\eqref{eq:constraint 1} and~\eqref{eq:constraint 2} at the same time. This completes the second part of the proof.

Finally, it is known from Lemma~\ref{lem:H2_norm} that~$|\Lambda(A_r)| < 1$, and~$|\Lambda(\mathcal{A}_{ijh})|<1$ are necessary and sufficient conditions for the feasibility of~\eqref{eq:constraint 3} and~\eqref{eq:constraint 4}, resp. 
Recalling the definition of~$\mathcal{A}_{ijh}$ in~\eqref{eq:closed-loop and filter}, ~$|\Lambda(\mathcal{A}_{ijh})|<1$ if and only if~$|\Lambda(A_r)| < 1$ and~$~|\Lambda(A^{cl}_{ih})|<1$.
This completes the proof.
\end{proof}

\subsection{Proofs of probabilistic certificates}
We introduce the following lemma to prove Theorem~\ref{thm:Prob certificate}.
\begin{lemma}[Linear transformation of sub-Gaussian signals]\label{lem:linear subG}
Suppose~$\T_{\omega r}$ is a transfer function from~$\omega$ to $r$ with the state-space realization~$(A,B,C)$, i.e.,~$r=\T_{\omega r} [\omega] = C(\mfq I-A)^{-1}B [\omega]$. If the input~$\omega$ is an iid sub-Gaussian signal with zero mean and parameter~$\zeta_{\omega}$, the output~$r$ is also sub-Gaussian with zero mean and the respective parameter~$\zeta_r = \| \T_{\omega r} \|_{\Ht} \zeta_{\omega}$.
\end{lemma}
\begin{proof}
From the linear system theory we know that~$r(k) - \mE[r(k)] = C\sum_{m=0}^{k-1} A^{k-1-m} B \omega(m)$. Then, for any constant~$\phi\in\mathbb{R}$ and a unit vector~$\nu$ with an appropriate dimension, we have
\begin{align}\label{eq:r-Er}
   \mE \left[\rme^{\phi \nu^{\top} (r(k) - \mE[r(k)])} \right] 
    = \mE\left[\rme^{\phi \nu^{\top} C\sum_{m=0}^{k-1} A^{k-1-m} B \omega(m)} \right] 
    = \prod_{m=0}^{k-1} \mE\left[\rme^{\phi \nu^{\top} C A^{k-1-m} B \omega(m)} \right],
\end{align}
Since~$\omega$ is sub-Gaussian, according to Lemma~\ref{lem:subG}, it holds that
\begin{align*}
    \mE \left[\rme^{\phi \nu^{\top} C A^{k-1-m} B \omega(m)}\right] \leq \rme^{\phi^2 \|\nu^{\top}\|^2_2 \|C A^{k-1-m}B \|_2^2 \zeta_{\omega}^2/2}.
\end{align*}
Recall that~$\|v\|_2=1$. Thus, equality~\eqref{eq:r-Er} satisfies
\begin{small}
\begin{align*}
    \mE \left[\rme^{\phi \nu^{\top}(r(k) - \mE[r(k)])} \right] 
    \leq \prod_{m=0}^{k-1} \rme^{\phi^2 \|C A^{k-1-m} B\|_2^2 \zeta_{\omega}^2/2} 
    =  \rme^{\phi^2 \sum_{m=0}^{k-1} \|C A^{k-1-m} B \|_2^2 \zeta_{\omega}^2/2}.
\end{align*}
\end{small}
By matrix norm definitions, we know~$\|A\|_2^2 \leq \textup{Trace}(A^{\top}A)$ for all real-valued matrix~$A$, and thus
\begin{align*}
    \mE \left[\rme^{\phi \nu^{\top}(r(k) - \mE[r(k)])} \right] 
    &\leq \rme^{\phi^2 \sum_{m=0}^{k-1} \textup{Trace} \left( C A^{k-1-m} B B^{\top} {A^{\top}}^{k-1-m} C^{\top} \right) \zeta_{\omega}^2/2} 
    \leq \rme^{\phi^2 \| \T_{\omega r} \|^2_{\Ht} \zeta_{\omega}^2/2},
\end{align*}
 where the last inequality follows from Parseval's Theorem and the~$\Ht$ norm definition.
\end{proof}
\begin{proof}[Proof of Theorem~\ref{thm:Prob certificate}]
The main idea builds on the probabilistic relation between the concentration of a random variable and its expectation. Since the noise~$\omega$ is sub-Gaussian, according to Lemma~\ref{lem:linear subG}, the matched residual~$r_{ij}$ is also sub-Gaussian with the parameter~$\zeta_{r_{ij}} = \left\| \T^{\Sij}_{\omega r_{ij}} \right\|_{\Ht}\zeta_{\omega} < \sqrt{\bar{\eta}_i}\zeta_{\omega}$. 
We first show that the performance guarantee~\eqref{eq:Prob guarantee} holds when~$|\mE[r_{ij}(k)]| \leq \mu\sqrt{\bar{\eta}_i}$. According to~\eqref{eq:threshold}, we have
\begin{align*}
    \varepsilon_i-\left|\mE[r_{ij}(k)]\right| \geq \varepsilon_i-\mu\sqrt{\bar{\eta}_i} = \zeta_{\omega}\sqrt{2 \ln{(2/ \beta)} \bar{\eta}_i}.
\end{align*}
Since it also holds that~$\left|r_{ij}(k)\right|-\left|\mE[r_{ij}(k)]\right| \leq \left| r_{ij}(k) - \mE[r_{ij}(k)] \right|$, we have 
\begin{align*}
    &\mPr \left[ \left|r_{ij}(k)\right| \leq \varepsilon_i \Big|
   \begin{bmatrix} \hsigma(k)  \\ \sigma(k) \end{bmatrix}
   =\begin{bmatrix} i  \\ j \end{bmatrix},  k \geq t_s \right] \notag\\
    =&\mPr \left[\left|r_{ij}(k)\right|-\left|\mE[r_{ij}(k)]\right| \leq \varepsilon_i -\left|\mE[r_{ij}(k)]\right| \Big|
   \begin{bmatrix} \hsigma(k)  \\ \sigma(k) \end{bmatrix}
   =\begin{bmatrix} i  \\ j \end{bmatrix},  k \geq t_s \right] \notag\\
    \geq  &\mPr \left[ \left| r_{ij}(k) - \mE[r_{ij}(k)] \right| \leq \zeta_{\omega}\sqrt{2 \ln{(2/\beta)}\bar{\eta}_i} \Big|
   \begin{bmatrix} \hsigma(k)  \\ \sigma(k) \end{bmatrix}
   =\begin{bmatrix} i  \\ j \end{bmatrix},  k \geq t_s \right] \notag\\
    \geq & 1- 2\rme^{ {-2\ln{(2/\beta)} \bar{\eta}_i\zeta^2_{\omega} \big/ \left(2\| \T_{\omega r_{ij}}^{\Sij} \|^2_{\Ht}\zeta^2_{\omega}\right)}} \geq 1- \beta,
\end{align*}
where the concentration inequality~\eqref{eq:Concentra Ine} in Lemma~\ref{lem:subG} is used to get the second inequality. This completes the first part of the proof.

Next, we show that~$|\mE[r_{ij}(k)]| \leq \mu\sqrt{\bar{\eta}_i}$ when~$k \geq t_s+T_{ij}$.
Let us incorporate the initial state~$x(t_s)$ into the expression of~$\mE[r_{ij}(k)]$, where~$x(t_s)$ is viewed as an input to the system that only has a nonzero value at~$t_s$.
According to the closed-loop dynamics~\eqref{eq:closed-loop dynamics}, for~$k = t_s+\Delta k$ where~$\Delta k \in [0,t^{\rm iso}_s)$, we have 
\begin{align}\label{eq:closed-loop dynamics with x_ts}
x(k+1) &= A^{cl}_{ij}x(k)+E_jd(k) +(W_j+B_jK_iD_j)\omega(k) + x(t_s), \notag\\
 y(k)  &= C_jx(k) + D_j\omega(k).
\end{align} 
We reformulate~\eqref{eq:closed-loop dynamics with x_ts} into the DAE format, which is
\begin{equation}\label{eq:DAE with x_ts}
\begin{bmatrix}
-\mfq I+A^{cl}_{ij}  &E_j  &I \\
C_j              &0    &0    
\end{bmatrix}
\begin{bmatrix}
x \\ d \\x(t_s)
\end{bmatrix} + 
L(\mfq) [y] + G_{ij}(\mfq) [\omega] = 0.
\end{equation}
Multiplying the left hand-side of~\eqref{eq:DAE with x_ts} by~$a^{-1}(\mfq)N_{ij}(\mfq)$ leads to
\begin{equation}\label{eq:r_ij(k)}
\begin{split}
r_{ij}  = \frac{N_{ij}(\mfq)L(\mfq)}{a(\mfq)} [y] 
          = -\frac{N_{ij}(\mfq)}{a(\mfq)}\begin{bmatrix}
-\mfq I+A^{cl}_{ij}  &E_j  &I \\
C_j              &0    &0    
\end{bmatrix}
\begin{bmatrix}
x \\ d \\x(t_s)
\end{bmatrix} 
 -\frac{N_{ij}(\mfq)G_{ij}(\mfq)}{a(\mfq)} [\omega].
\end{split}
\end{equation}   
Recall that~$N_{ij}(\mfq)H_{ij}(\mfq) = 0$ in Theorem~\ref{thm:minimization problem}. By substituting~$N_{ij}(\mfq) = \left[\hat{N}_{ij}(\mfq) ~\check{N}_{ij}(\mfq)\right]$ into~\eqref{eq:r_ij(k)}, we have
\begin{align*}
r_{ij} = -\frac{\hat{N}_{ij}(\mfq)}{a(\mfq)}x(t_s) -\frac{N_{ij}(\mfq)G_{ij}(\mfq)}{a(\mfq)} [\omega].
\end{align*}
Hence, the expected value of~$r_{ij}$ is
\begin{align*}
    \mE[r_{ij}] = -a^{-1}(\mfq)\hat{N}_{ij}(\mfq)\mE[x(t_s)].
\end{align*}

To compute~$T_{ij}$, following the idea of~\cite[Lemma 3.4]{van2022multiple}, we transform~$-a^{-1}(\mfq)\hat{N}_{ij}(\mfq)$ to its Jordan canonical form denoted by~$(\mathsf{A},\mathsf{B}_{ij},\mathsf{C})$. The transfer function~$-a^{-1}(\mfq)\hat{N}_{ij}(\mfq)$ can be expanded as
\begin{align*}
-\frac{\hat{N}_{ij}(\mfq)}{a(\mfq)}  = 
&\left[-\frac{\sum^{d_N}_{m=0}\hat{N}_{ij,m}(1)\mfq^{m}}{a(\mfq)}, \dots,   -\frac{\sum^{d_N}_{m=0}\hat{N}_{ij,m}(n_x)\mfq^{m}}{a(\mfq)} \right],
\end{align*}
Recall that $a(\mfq) = \prod^{d_N+1}_{\ell=1}(\mfq-\lambda_\ell)$. The factorization of the~$h$-th element of~$-a^{-1}(\mfq)\hat{N}_{ij}(\mfq)$ is
\begin{equation*}
-\frac{\sum^{d_N}_{m=0}\hat{N}_{ij,m}(h)\mfq^{m}}{a(\mfq)}
= \sum_{\ell=1}^{d_N+1} \frac{b_{ij,\ell h}}{\mfq-\lambda_{\ell}},
\end{equation*} 
where~$b_{ij,\ell h} = -\frac{\sum^{d_N}_{m=0}\hat{N}_{ij,m}(h) \lambda_{\ell}^{m}}{\prod_{\tilde{\ell} \neq \ell} (\lambda_{\tilde{\ell}} - \lambda_\ell)}$. The Jordan canonical form of~$-a^{-1}(\mfq)\sum^{d_N}_{m=0}\hat{N}_{ij,m}(h)\mfq^{m} $ is denoted by~$(\mathsf{A}_{h},\mathsf{B}_{ij,h},\mathsf{C}_{h})$, where
\begin{align*}
    &\mathsf{A}_{h} = \text{diag}([\lambda_1,\dots,\lambda_{d_N+1}]), 
    ~\mathsf{B}_{ij,h} = [b_{ij,1h},\dots,b_{ij,(d_N+1)h}]^{\top},
    ~\mathsf{C}_{h} = [1,\dots,1].
\end{align*}
According to the superposition property of linear systems, we have
$\mathsf{A} = \text{diag}([\lambda_1,\dots,\lambda_{d_N+1}]), 
~\mathsf{B}_{ij} = [\mathsf{B}_{ij,1},\dots,\mathsf{B}_{ij,n_x}],
~\mathsf{C} = [1,\dots,1]$.
With the state-space description,~$\mE[r_{ij}(k)]$ can be written as
\begin{align*}
\mE[r_{ij}(k)] &= \mathsf{C} \mathsf{A}^{\Delta k} \mE[\bar{x}_{ij}(t_s)] 
+ \mathsf{C}\sum_{m=0}^{\Delta k -1} \mathsf{A}^{\Delta k-1-m} \mathsf{B}_{ij} \mE[x(t_s)] \\
 &=\mathsf{C} \mathsf{A}^{\Delta k} \mE[\bar{x}_{ij}(t_s)] + \mathsf{C} \mathsf{A}^{\Delta k-1} \mathsf{B}_{ij} \mE[x(t_s)]
\end{align*}  
where~$\bar{x}_{ij}(t_s)$ is the filter state.
Since~$\mathsf{A}$ is a diagonal matrix, we have~$\|\mathsf{A} \|_2 = \lambda_{\max}$. Based on the triangle property of norms,~$\vert \mE[r_{ij}(k)] \vert$ is bounded by
\begin{align*} 
| \mE[r_{ij}(k)] | 
&\leq \|\mathsf{C}\|_2 \| \mathsf{A} \|_2^{\Delta k} \|\mE[\bar{x}_{ij}(t_s)]\|_2
+ \|\mathsf{C}\|_2 \|\mathsf{A} \|_2^{\Delta k -1} \|\mathsf{B}_{ij}\|_2 \|\mE[x(t_s)]\|_2 \\
&\leq \sqrt{d_N+1} (1 + \lambda_{\max}^{-1} \|\mathsf{B}_{ij}\|_2) \|\mE[\X_{ij}(t_s)]\|_2 \lambda_{\max}^{\Delta k} \\
&= \psi_{ij}\left(\Fij,\X_{ij}(t_s)\right)\lambda_{\max}^{\Delta k}.
\end{align*}
By setting~$\mu\sqrt{\bar{\eta}_i} \geq \psi_{ij}\left(\Fij,\X_{ij}(t_s)\right)\lambda_{\max}^{\Delta k}$, we arrive at
\begin{align*}
 \Delta k  \geq T_{ij} = \left\lceil \log_{\lambda_{\max}} \frac{\mu\sqrt{\bar{\eta}_{i}}}{\psi_{ij}\left(\Fij,\X_{ij}(t_s)\right)}  \right\rceil .
\end{align*}
That completes the proof.
\end{proof}

\section{Illustrative Examples}\label{sec:simulation}

In this section, we consider a numerical example and a practical application on building radiant systems to illustrate the effectiveness of the proposed diagnosis scheme. 

\subsection{Numerical results}
Consider a switched system with three linear subsystems. The system matrices are
\begin{equation*}
\begin{split}
&A_1 = \begin{bmatrix} 0.5 &0 \\ 0 &-0.4\end{bmatrix}, A_2 = \begin{bmatrix} 0.5 &-0.2\\ 0 &-0.4\end{bmatrix}, A_3 = \begin{bmatrix} -0.5 &0 \\ 0.1 &-0.4 \end{bmatrix}, B_1 = \begin{bmatrix} 0 \\ 1 \end{bmatrix},\\
&B_2 = \begin{bmatrix} 1 \\ 1\end{bmatrix},
B_3 = \begin{bmatrix} 1 \\ 0 \end{bmatrix}, E_1=E_2=E_3 = \begin{bmatrix} 1 \\1\end{bmatrix}, W_1=W_2=W_3=0, \\
&C_1 = C_3 =\begin{bmatrix} 1 &0 \\ 0 &1\end{bmatrix}, C_2 = \begin{bmatrix} 1 &0\\ 0 &0\end{bmatrix}, D_1=D_2=D_3 = \begin{bmatrix} 0.01 &0 \\ 0.01 &-0.01 \end{bmatrix}.
\end{split}
\end{equation*}
The controller gains are~$K_1=[-0.0395 ~-0.0741]$,~$K_2=[-0.0648 ~0.0510]$, and~$K_3 = [-0.0420 ~0.0326]$.
We set the degree of the filters~$d_N = 1$, the denominator~$a(\mfq)=(\mfq+0.1)(\mfq+0.2)$. 
The reference signal is set as~$\bar{d} = 0.5$. The parameter of the iid sub-Gaussian noise is~$1$.
The filters are constructed by using the approach proposed in Theorem~\ref{thm:minimization problem} and Proposition~\ref{prop:LMI Appro}. 
We solve the optimization problems by YALMIP toolbox~\cite{2004YALMIP}. 
The thresholds are computed according to \eqref{eq:threshold} where the reliability level~$\beta = 0.05$ and~$\mu = 0.5$. Thus, the thresholds are~$\varepsilon_{1} = 0.18,~\varepsilon_{2} = 0.16,~\varepsilon_{3} = 0.12$. The waiting time~$\tau_i$ for~$i\in\{1,2,3\}$ computed by~\eqref{eq:tauj} are~$\tau_1=7, \tau_2=6, \tau_3=7$.
To cover all the scenarios, we set the switching sequence as:~$1 \rightarrow 2 \rightarrow 3 \rightarrow 1 \rightarrow 3 \rightarrow 2 \rightarrow 1$.

\begin{figure}[t]
\centering
\subfigure[$|r_{1h}|$ with~$\mathsf{M}^{50}_{12}$] {
\includegraphics[scale=0.48]{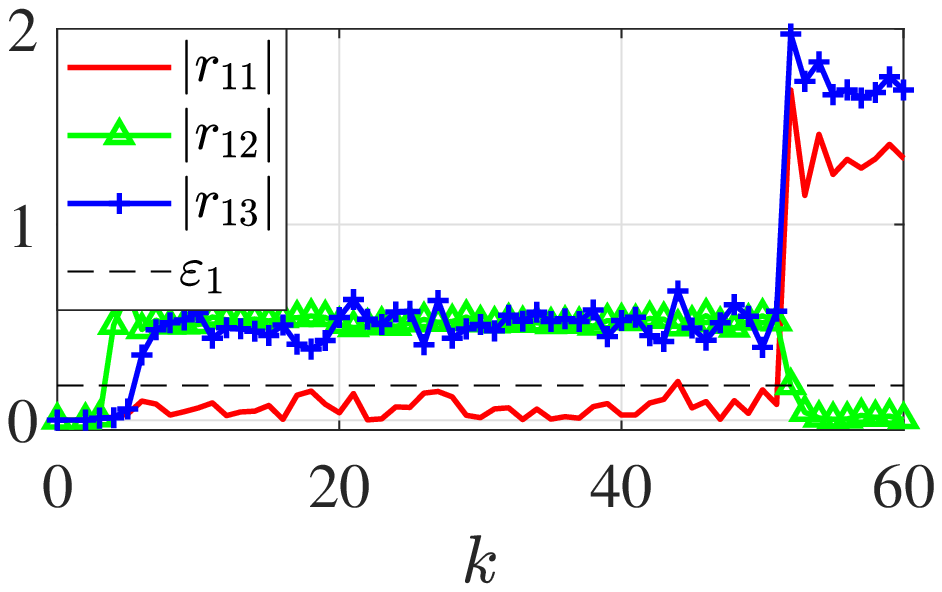}\label{fig:residuals.a}
}
\hspace{-0.5cm}
\subfigure[ $|r_{1h}|$ with~$\mathsf{M}^{200}_{13}$]{
\includegraphics[scale=0.48]{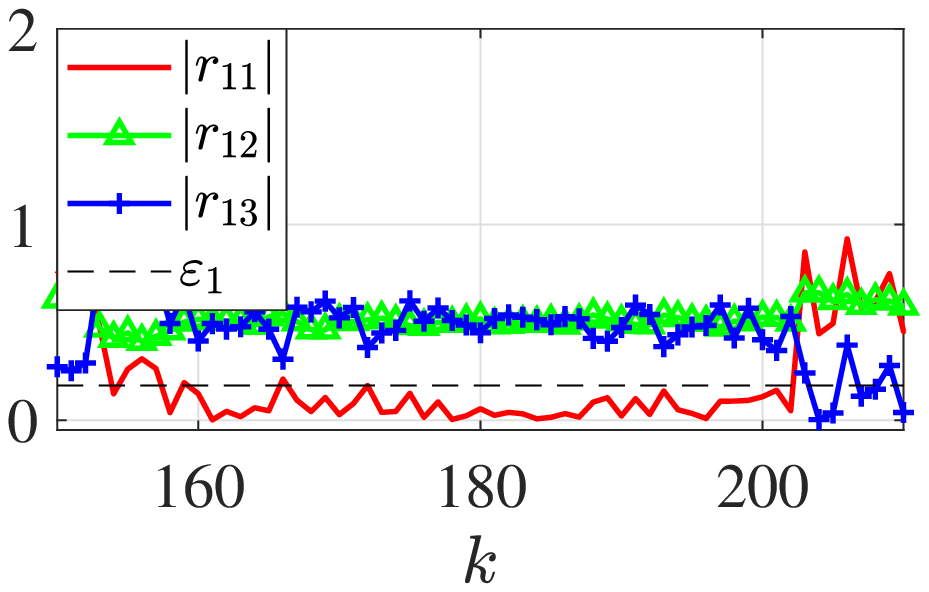}
}
\hspace{-0.5cm}
\subfigure[ $|r_{2h}|$ with~$\mathsf{M}^{300}_{21}$]{
\includegraphics[scale=0.48]{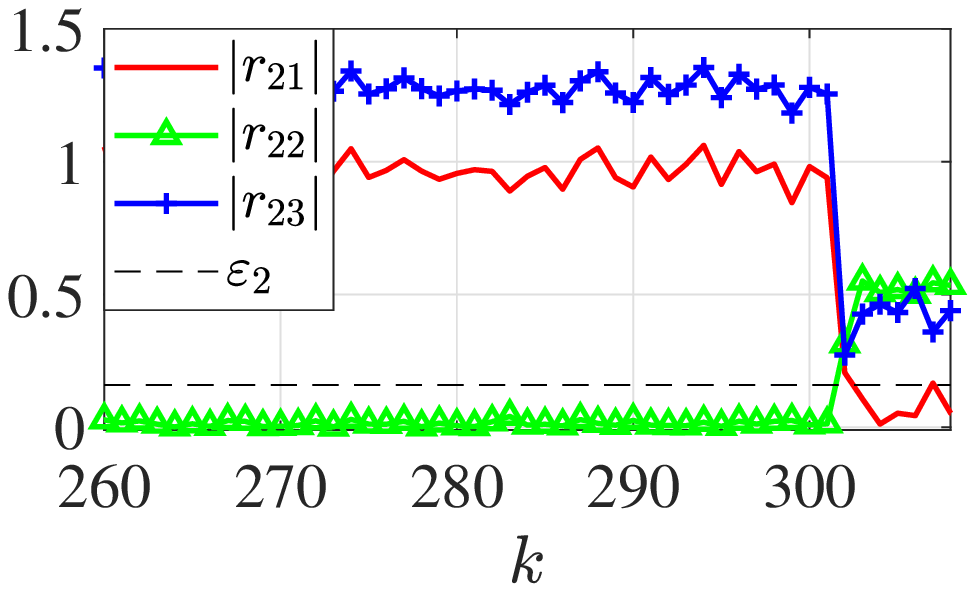}\label{fig:residuals.c}
}
\hspace{-0.5cm}
\subfigure[$|r_{2h}|$ with~$\mathsf{M}^{100}_{23}$]{
\includegraphics[scale=0.48]{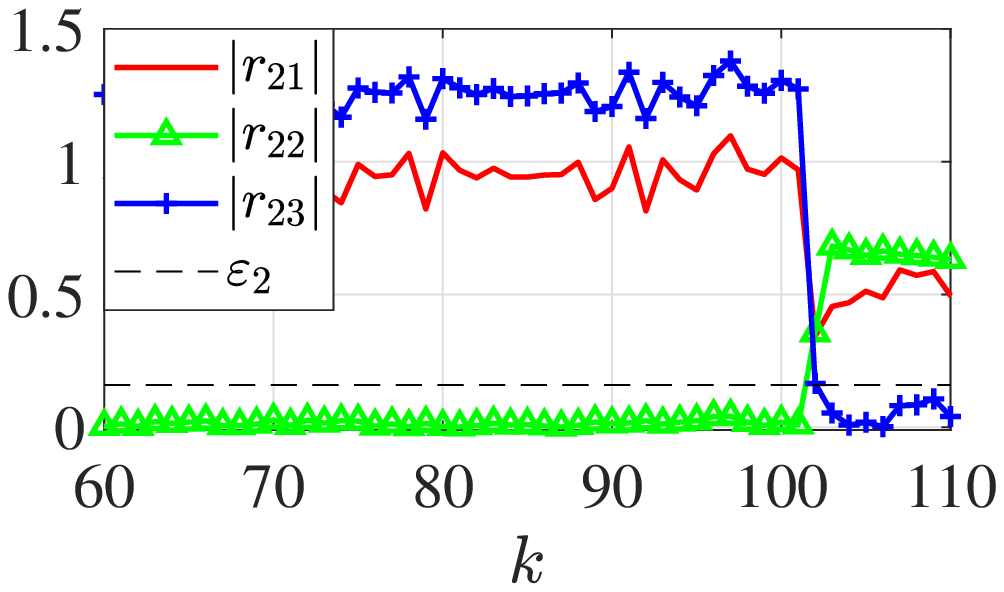}\label{fig:residuals.d}
}
\hspace{-0.5cm}
\subfigure[$|r_{3h}|$ with~$\mathsf{M}^{150}_{31}$]{
\includegraphics[scale=0.48]{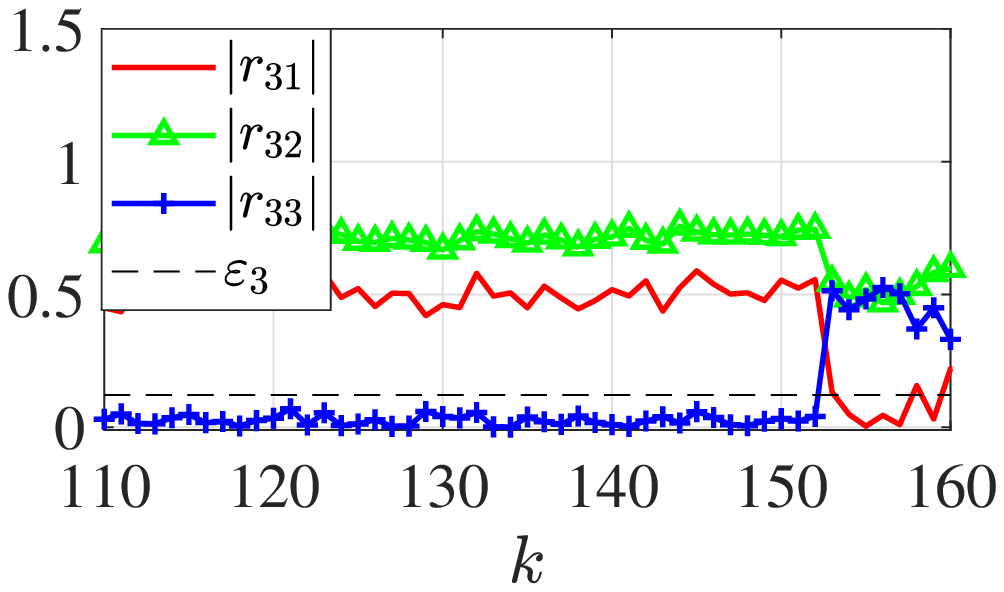}
}
\hspace{-0.5cm}
\subfigure[$|r_{3h}|$ with~$\mathsf{M}^{250}_{32}$]{
\includegraphics[scale=0.48]{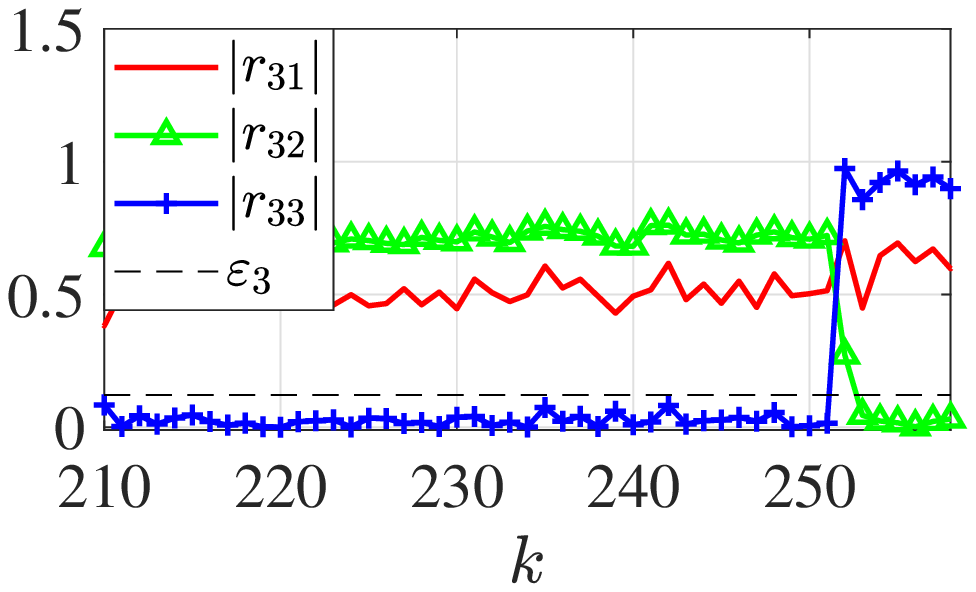}
}
\vspace{-0.2cm}
\caption{Residuals behavior under different scenarios: Let $\mathsf{M}^{k}_{ij}$ stand for a system transition from $i$ to $j$ at time $k$.}\label{fig:residuals}
\end{figure}   

\begin{figure}[t]
\center
\includegraphics[scale=0.6]{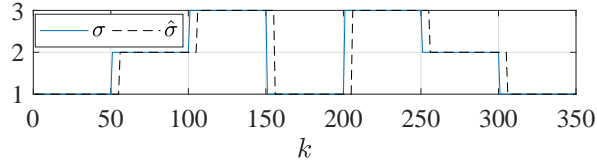}
\vspace{-0.2cm}
\caption{Diagnosis result of the whole process}\label{fig:diagnosis result}
\end{figure} 

\begin{figure} [h]
\centering
\includegraphics[scale=0.6]{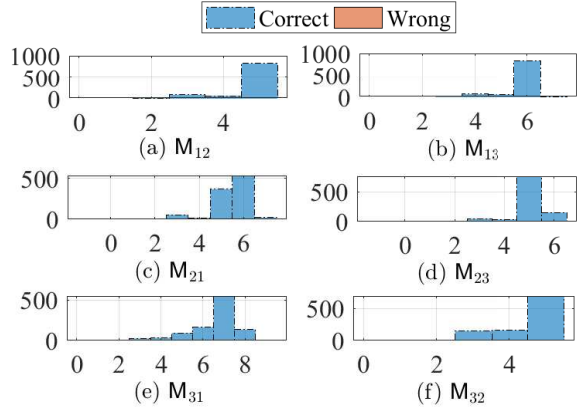}
\vspace{-0.2cm}
\caption{Distribution of the diagnosis time for each scenario} \label{fig:Diagnosis time distribution}
\end{figure}

Figure~\ref{fig:residuals} depicts the residuals behavior under different scenarios. 
Here, we only analyze~$r_{1h}$ for~$h \in \{1,2,3\}$ with the transition~$\mathsf{M}^{50}_{12}$ shown in Figure~\ref{fig:residuals.a}, because the rest are similar. 
Since the initial status of the closed-loop system is~$\mathsf{S}_{11}$, the absolute value~$|r_{11}(k)|$ remains below~$\varepsilon_{1}$ until transition happens at~$k=50$. 
The other two residuals $r_{12}$ and~$r_{13}$ first reach their corresponding steady values and then oscillate around the steady values because of the noise.
The matched residual~$r_{11}$ and the unmatched residuals~$r_{12}$ and~$r_{13}$ are separated. 
After the transition~$\mathsf{M}^{50}_{12}$ happens at~$k=50$,~$|r_{11}(k)|$ exceeds the threshold~$\varepsilon_{1}$ immediately such that the switching is detected. 
Then,~$|r_{12}(k)|$ reaches~$\varepsilon_{1}$ at about~$k = 53$ while the other two residuals are above~$\varepsilon_{1}$. As a result, active mode 2 is determined.
Figure~\ref{fig:diagnosis result} shows the diagnosis result of the whole process, where the switching signal is correctly estimated.
 
We execute the experiment 1000 times for each switching scenario to obtain the distributions of the diagnosis time and the probability of wrong detection. The results are shown in Figure~\ref{fig:Diagnosis time distribution}.
The average diagnosis time (ADT) and the wrong detection probability (WDP) are presented in Table~\ref{tab:transition}. 
We compute the estimated matched time~$T_{ij}$ based on~\eqref{eq:T_ij}.
From Table~\ref{tab:transition}, the estimated matched time estimates the average diagnosis time well, and the wrong detection probability is low.  
\begin{table}[h]
\small
\centering
\caption{Average diagnosis time and wrong detection probability when~$\mu=0.5$ and~$\beta=0.05$}\label{tab:transition}
{\begin{tabular}{l c c c c c c}
\toprule
 Transition &$\mathsf{M}_{12}$ &$\mathsf{M}_{13}$ &$\mathsf{M}_{21}$ &$\mathsf{M}_{23}$ &$\mathsf{M}_{31}$ &$\mathsf{M}_{32}$ \\
\midrule
  ADT         &5       &6      &6    &6       &7     &5 \\
  $T_{ij}$    &5       &7      &7    &7       &7      &5 \\
  WDP  &0   &0      &0.002    &0.003   &0  &0 \\
\bottomrule
\end{tabular}} 
\end{table}

\subsection{Building radiant systems}
In this section, a building radiant system is considered. 
We adopt the example from~\cite{2018Guaranteed}, where the building with four rooms of the same size is equipped with a radiant system with two pumps. 
Moreover, we compare the model invalidation approach proposed in~\cite{2018Guaranteed} with our approach.

\subsubsection{System model description}
The radiant system can be modeled by the following equations
\begin{small}
\begin{align*}
&C_{c,1}\dot{T}_{c,1} = K_{c,1}(T_1-T_{c,1}) + K_{c,3}(T_3-T_{c,1})+ K_{w,1}(T_{w,1} - T_{c,1}),\\
&C_{c,2}\dot{T}_{c,2} = K_{c,2}(T_2-T_{c,2}) + K_{c,4}(T_4-T_{c,2})+ K_{w,2}(T_{w,2} - T_{c,2}),\\
&C_1\dot{T}_1 = K_{c,1}(T_{c,1}-T_1) + K_1(T_a-T_1)+K_{12}(T_2-T_1) +K_{13}(T_3-T_1),\\ 
&C_2\dot{T}_2 = K_{c,2}(T_{c,2}-T_2) + K_2(T_a-T_2)+K_{12}(T_1-T_2)+K_{24}(T_4-T_2),\\
&C_3\dot{T}_3 = K_{c,1}(T_{c,1}-T_3) + K_3(T_a-T_3)+K_{13}(T_1-T_3)+K_{34}(T_4-T_3),\\
&C_2\dot{T}_4 = K_{c,2}(T_{c,2}-T_4) + K_4(T_a-T_4)+K_{24}(T_2-T_4)+K_{34}(T_3-T_4),
\end{align*}
\end{small}
where the temperatures of two cores in the radiant system are denoted by~$T_{c,i}$ for~$i \in \{1,2\}$. 
The temperature of the supply water is denoted by~$T_{w,i}$.
The ambient air temperature is denoted by~$T_a$. The air temperature of room~$i$ for~$i \in \{1,2,3,4\}$ is denoted by~$T_i$. The thermal conductance between~$T_i$ and~$T_a$ is denoted by~$K_{i}$. The thermal conductance between~$T_{c,i}$ and~$T_i$ is denoted by~$K_{c,i}$. The thermal conductance between room~$i$ and~$j$ is denoted by~$K_{ij}$. 
The piping thermal conductance between~$T_{c,i}$ and~$T_{w,i}$ is denoted by~$K_{w,i}$.
The thermal capacitance of room~$i$ and core~$i$ is denoted by~$C_i$ and~$C_{c,i}$, resp.
Assume that the constant flow of pumps is known. Each pump supplies water to the water pipe and is connected to a valve to adjust the constant flow. 
The system state consists of the temperatures of the four rooms and the two cores.
Suppose both pumps are on. 
The values of the parameters are the same as that in~\cite{2018Guaranteed}. 
The above equations can be written into the state-space form
\begin{align}\label{eq:radiant system}
\dot{x}_T &= A_{rad,1}x_T+E_{rad,1}T_d, \notag \\
y &= C_{rad,1}x_T+\omega,
\end{align}
where~$x_T = [T_{c,1},T_{c,2},T_1,T_2,T_3,T_4]^{\top}$,~$T_d =[T_{w,1},T_{w,2},T_a]^{\top}$ is the constant input (or reference signal). Matrices~$A_{rad,1}$ and~$E_{rad,1}$ are obtained from the above equations. The matrix~$C_{rad,1} = \text{diag}([0,0,1,1,1,1])$ indicates the measured temperatures. 
Assume that there is an uncertainty~$\nu$ in~$T_a$ due to small changes (i.e.,~$T_a = 10 + \nu$ where~$\nu$ is Gaussian noise with mean~$0$ and variance 0.1).
The measurement noise denoted by~$\omega$ is Gaussian noise with mean~$0$ and variance 0.01.
The discrete-time model of the radiant system~\eqref{eq:radiant system} is obtained with a sampling time of 5 min. Let~$(A^d_{rad,1},E^d_{rad,1},C_{rad,1})$ represents the fault-free discrete-time model of the system. 

\subsubsection{Faulty modes}
The normal functions of the valves and temperature measurement sensors are impaired in the faulty modes. 
Specifically, when there is a fault in the valve, we assume that the valve is stuck in the middle and does not respond to commands. Since the fault cuts the heat transfer in half, the fault is modeled with a change in the heat conductance parameter, i.e.,~$K_{w,1} \rightarrow K_{w,1}/2$ in~$A_{rad,1}$ and~$E_{rad,1}$. 
The sensor failures result in inaccurate measurements of the temperature. We change the corresponding entry in~$C_{rad,1}$ to model the sensor fault, i.e.,~$1 \rightarrow 0.9$ .
Here, two faulty modes are considered.
The first faulty mode is denoted by~$(A^d_{rad,2},E^d_{rad,2},C_{rad,2})$, where faults occur in the second pump and the sensor measuring~$T_{1}$. As a result,~$K_{w,2}$ decreases to~$K_{w,2}/2$ and~$C_{rad,2}=\text{diag}([0,0,0.9,1,1,1])$. 
The second faulty mode is denoted by~$(A^d_{rad,3},E^d_{rad,3},C_{rad,3})$, where just one fault occurs in the first pump.  
Note that the second faulty mode is more incipient than the first one 
because the outputs do not change dramatically.
The matched residual of~$(A^d_{rad,i},E^d_{rad,i},C_{rad,i})$ is defined as~$r_{i}$ for~$i \in \{1,2,3\}$.

\subsubsection{Filter design and model invalidation approach}
Note that there is no control signal in the radiant system~\eqref{eq:radiant system}. Thus, we only need to design three filters corresponding to the three modes. The degree of the filters is set as~$d_N=3$. The filters are then constructed based on Theorem~\ref{thm:minimization problem} and Proposition~\ref{prop:LMI Appro}.
The idea of the model invalidation approach proposed in~\cite{2018Guaranteed} is that, given the input and output data, detect the transitions by checking the feasibility of a mixed-integer linear programming problem. 
Since the example we adopt here has only one healthy mode, the MILP problem degenerates into the following linear programming problem.
\begin{small}
\begin{align}\label{eq:invalidation approach}
\text{Find} &~{\bf{x}}(k),~\boldsymbol{\nu}(k),~\boldsymbol{\omega}(k),~\forall k \in \{0,1,\dots,T-1\} \notag \\
\text{s.t.} &\left\{
\begin{array}{l}
{\bf{x}}(k+1)-A_{rad,1}{\bf{x}}(k)-E_{rad,1}(T_d +[0,0,\boldsymbol{\nu}(k)]^{\top}) = 0,\\
y(k)-C_{rad,1}{\bf{x}}(k)-\boldsymbol{\omega}(k)=0,\\
X_l \leq {\bf{x}}(k) \leq X_u, ~V_l \leq \boldsymbol{\nu}(k) \leq V_u,\\
W_l \leq \boldsymbol{\omega}(k) \leq W_u.
\end{array} \right.
\end{align} 
\end{small}
where the ranges of~${\bf{x}}(k)$,~$\nu(k)$ and~$\omega(k)$ are set as~$15 \leq \|{\bf{x}}\|_\infty \leq 19$,~$-0.3 \leq \|\boldsymbol{\nu}\|_\infty \leq 0.3$ and~$-0.03 \leq \|\boldsymbol{\omega}\|_\infty  \leq 0.03$, resp.
The positive integer~$T$ is derived from the definition \textit{T-Detectability} in~\cite{2018Guaranteed}. 
It represents the number of steps that a faulty model needs to generate an abnormal trajectory. 
We refer readers to~\cite{2018Guaranteed} for more details about the computation method of~$T$.  

\subsubsection{Results}
In the first case, we suppose the first faulty mode occurs at~$k = 20$. The diagnosis results are presented in Figure~\ref{fig:faultymode1}. Figure~\ref{fig:faultymode1}(a) shows the changes in the measured temperatures. 
The temperature~$T_1$ drops significantly due to sensor failure, and other measured temperatures also change slightly because of the fault in pump 2.  
Figure~\ref{fig:faultymode1}(b) shows the changes in the residuals and the feasibility of the invalidation problem~\eqref{eq:invalidation approach}.
One can see that~$r_1$ crosses the threshold at~$k=21$, and thus the fault is detected immediately after the faults happen. At~$k = 23$, the matched residual~$r_2$ reaches the threshold. Thus, the faulty mode is determined.
Meanwhile, the problem~\eqref{eq:invalidation approach} becomes infeasible at~$k=21$, which means the faults are detected by the model invalidation method as well.
In the second case, we suppose the second faulty mode happens at~$k=20$. 
One can see from Figure~\ref{fig:faultymode1}(c) that the changes in the measured temperatures are slight. This poses a challenge to the diagnosis task. Figure~\ref{fig:faultymode1}(d) shows the changes in the residuals and the feasibility of~\eqref{eq:invalidation approach}. Note that~$r_1$ crosses the threshold at~$k=22$. Hence, the fault is detected. 
Then, the matched residual~$r_3$ reaches the threshold at~$k=24$ such that the second faulty mode is determined. 
As a comparison, the invalidation problem is always feasible during the whole process, which means that the invalidation approach fails to detect the fault in the second case.  

\begin{figure}[t]
\centering
\subfigure[Faulty mode 1: Temperatures]{
\includegraphics[scale=0.6]{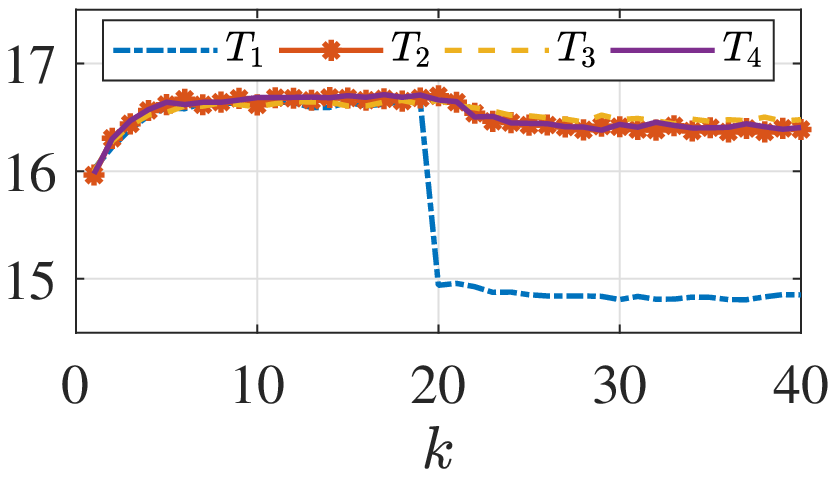}
}
\hspace{-2mm}
\subfigure[Faulty mode 1: Residuals]{
\includegraphics[scale=0.6]{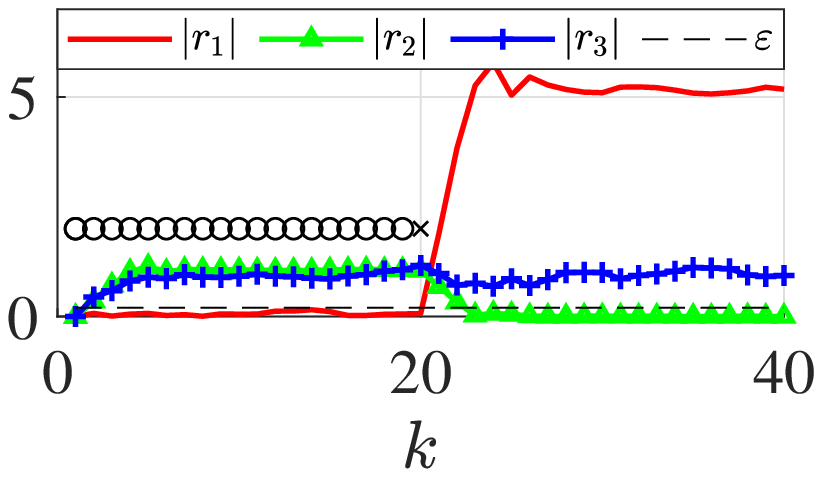}
}
\vspace{-0.25cm}

\centering
\subfigure[Faulty mode 2: Temperatures]{
\includegraphics[scale=0.6]{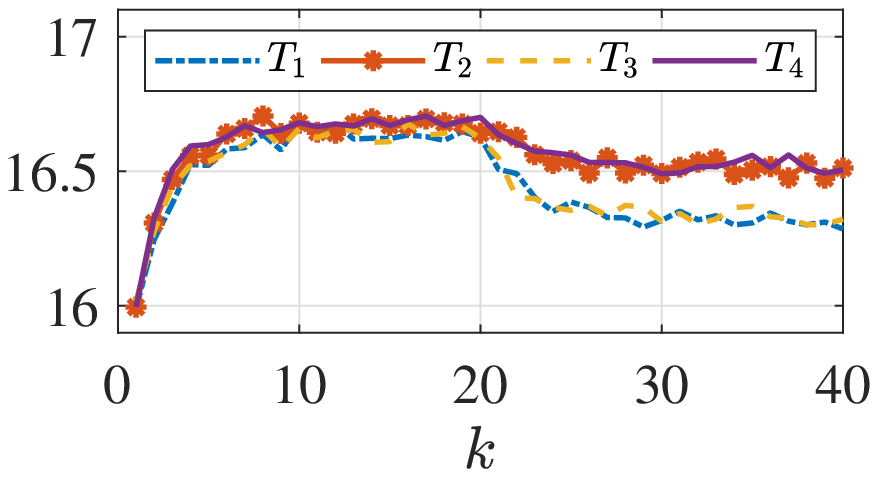}
}
\hspace{-2mm}
\subfigure[Faulty mode 2: Residuals]{
\includegraphics[scale=0.6]{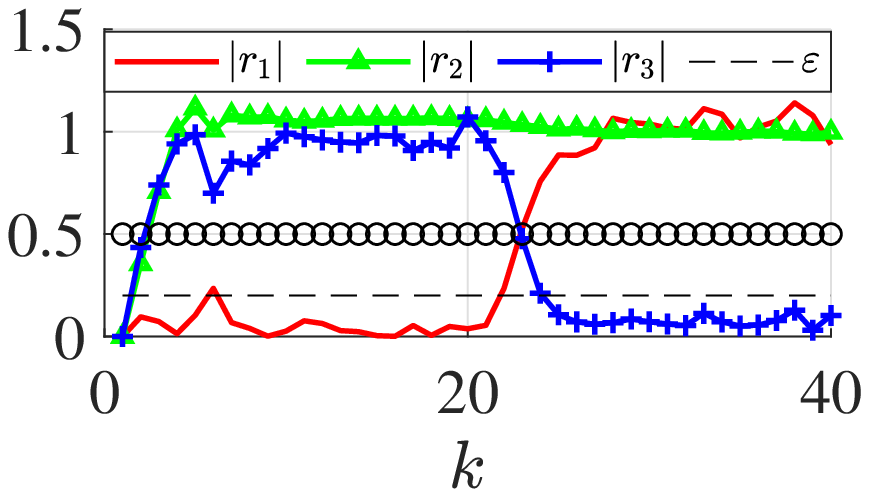}
}
\vspace{-0.25cm}
\caption{Simulation results with faulty modes happen at~$k=20$. The symbols~$\circ$ and $\times$ indicate the feasible and infeasible status of~\eqref{eq:invalidation approach}}\label{fig:faultymode1}
\end{figure}

\section{Conclusion and future directions}\label{sec:conslusion}
In this paper, we propose a diagnosis scheme to detect the active mode of discrete-time, switched affine systems in the presence of measurement noise and asynchronous switching. Based on an integration of residual generation and~$\Ht$-norm approaches, the design of an optimal bank of filters is formulated into a tractable optimization problem in which the noise contribution to the residuals is minimized. With the filters designed by the optimization problem, the diagnosis thresholds are determined which provide probabilistic false-alarm guarantees on the mode detection performance. Simulation results of a numerical example and a building radiant system show the effectiveness of the proposed approach. 
As future work, the first research direction is to combine the proposed approach with the active fault diagnosis method to deal with the unknown disturbance. One can design certain input sequences such that the unmatched residuals are separated from the matched residual with guaranteed probability. 
Note that the switching delay between the active mode and its corresponding controller is stochastic because of the stochastic noise. As a result, the second research direction would be focused on the impacts of the stochastic delay on the stability of asynchronously switched systems.

\bibliographystyle{elsarticle-num}       
\bibliography{mybibfile} 

\begin{thebibliography}{10}
\expandafter\ifx\csname url\endcsname\relax
  \def\url#1{\texttt{#1}}\fi
\expandafter\ifx\csname urlprefix\endcsname\relax\def\urlprefix{URL }\fi
\expandafter\ifx\csname href\endcsname\relax
  \def\href#1#2{#2} \def\path#1{#1}\fi

\bibitem{venkatasubramanian2003review}
V.~Venkatasubramanian, R.~Rengaswamy, K.~Yin, S.~N. Kavuri, A review of process
  fault detection and diagnosis: Part~{I}: Quantitative model-based methods,
  Computers \& Chemical Engineering 27~(3) (2003) 293--311.

\bibitem{zolghadri2012advanced}
A.~Zolghadri, Advanced model-based {FDIR} techniques for aerospace systems:
  Today challenges and opportunities, Progress in Aerospace Sciences 53 (2012)
  18--29.

\bibitem{weimer2013parameter}
J.~Weimer, J.~Araujo, M.~Amoozadeh, S.~A. Ahmadi, H.~Sandberg, K.~H. Johansson,
  Parameter-invariant actuator fault diagnostics in cyber-physical systems with
  application to building automation, in: Control of Cyber-Physical Systems,
  Springer, 2013, pp. 179--196.

\bibitem{bako2011identification}
L.~Bako, Identification of switched linear systems via sparse optimization,
  Automatica 47~(4) (2011) 668--677.

\bibitem{ohlsson2013identification}
H.~Ohlsson, L.~Ljung, Identification of switched linear regression models using
  sum-of-norms regularization, Automatica 49~(4) (2013) 1045--1050.

\bibitem{ackerson1970state}
G.~Ackerson, K.~Fu, On state estimation in switching environments, IEEE
  Transactions on Automatic Control 15~(1) (1970) 10--17.

\bibitem{lin2009stability}
H.~Lin, P.~J. Antsaklis, Stability and stabilizability of switched linear
  systems: a survey of recent results, IEEE Transactions on Automatic control
  54~(2) (2009) 308--322.

\bibitem{yuan2018novel}
S.~Yuan, L.~Zhang, B.~De~Schutter, S.~Baldi, A novel lyapunov function for a
  non-weighted ${L}_2$ gain of asynchronously switched linear systems,
  Automatica 87 (2018) 310--317.

\bibitem{gao2015survey}
Z.~Gao, C.~Cecati, S.~X. Ding, A survey of fault diagnosis and fault-tolerant
  techniques—part {I}: Fault diagnosis with model-based and signal-based
  approaches, IEEE Transactions on Industrial Electronics 62~(6) (2015)
  3757--3767.

\bibitem{beard1971failure}
R.~V. Beard, Failure accomodation in linear systems through
  self-reorganization, Ph.D. thesis, Massachusetts Institute of Technology
  (1971).

\bibitem{henry2005design}
D.~Henry, A.~Zolghadri, Design and analysis of robust residual generators for
  systems under feedback control, Automatica 41~(2) (2005) 251--264.

\bibitem{chow1984analytical}
E.~Chow, A.~Willsky, Analytical redundancy and the design of robust failure
  detection systems, IEEE Transactions on Automatic Control 29~(7) (1984)
  603--614.

\bibitem{frisk2001minimal}
E.~Frisk, M.~Nyberg, A minimal polynomial basis solution to residual generation
  for fault diagnosis in linear systems, Automatica 37~(9) (2001) 1417--1424.

\bibitem{nyberg2006residual}
M.~Nyberg, E.~Frisk, Residual generation for fault diagnosis of systems
  described by linear differential-algebraic equations, IEEE Transactions on
  Automatic Control 51~(12) (2006) 1995--2000.

\bibitem{seliger1991fault}
R.~Seliger, P.~M. Frank, Fault-diagnosis by disturbance decoupled nonlinear
  observers, in: the 30th IEEE Conference on Decision and Control, 1991, pp.
  2248--2253.

\bibitem{benosman2010survey}
M.~Benosman, A survey of some recent results on nonlinear fault tolerant
  control, Mathematical Problems in Engineering 2010 (2010).

\bibitem{boem2011distributed}
F.~Boem, R.~M. Ferrari, T.~Parisini, Distributed fault detection and isolation
  of continuous-time non-linear systems, European Journal of Control 17~(5-6)
  (2011) 603--620.

\bibitem{ferrari2011distributed}
R.~M. Ferrari, T.~Parisini, M.~M. Polycarpou, Distributed fault detection and
  isolation of large-scale discrete-time nonlinear systems: An adaptive
  approximation approach, IEEE Transactions on Automatic Control 57~(2) (2011)
  275--290.

\bibitem{esfahani2015tractable}
P.~Mohajerin~Esfahani, J.~Lygeros, A tractable fault detection and isolation
  approach for nonlinear systems with probabilistic performance, IEEE
  Transactions on Automatic Control 61~(3) (2015) 633--647.

\bibitem{pan2021dynamic}
K.~Pan, P.~Palensky, P.~Mohajerin~Esfahani, Dynamic anomaly detection with
  high-fidelity simulators: A convex optimization approach, IEEE Transactions
  on Smart Grid 13~(2) (2021) 1500--1515.

\bibitem{halimi2014model}
M.~Halimi, G.~Mill{\'e}rioux, J.~Daafouz, Model-based modes detection and
  discernibility for switched affine discrete-time systems, IEEE Transactions
  on Automatic Control 60~(6) (2014) 1501--1514.

\bibitem{kusters2018switch}
F.~K{\"u}sters, S.~Trenn, Switch observability for switched linear systems,
  Automatica 87 (2018) 121--127.

\bibitem{frank1990fault}
P.~M. Frank, Fault diagnosis in dynamic systems using analytical and
  knowledge-based redundancy: A survey and some new results, Automatica 26~(3)
  (1990) 459--474.

\bibitem{cocquempot2004fault}
V.~Cocquempot, T.~El~Mezyani, M.~Staroswiecki, Fault detection and isolation
  for hybrid systems using structured parity residuals, in: the 5th Asian
  Control Conference, Vol.~2, IEEE, 2004, pp. 1204--1212.

\bibitem{wang2007adaptive}
D.~Wang, K.~Y. Lum, Adaptive unknown input observer approach for aircraft
  actuator fault detection and isolation, International Journal of Adaptive
  Control and Signal Processing 21~(1) (2007) 31--48.

\bibitem{mincarelli2016uniformly}
D.~Mincarelli, A.~Pisano, T.~Floquet, E.~Usai, Uniformly convergent sliding
  mode-based observation for switched linear systems, International Journal of
  Robust and Nonlinear Control 26~(7) (2016) 1549--1564.

\bibitem{zhang2019sliding}
Z.~Zhang, S.~Li, H.~Yan, Q.~Fan, Sliding mode switching observer-based actuator
  fault detection and isolation for a class of uncertain systems, Nonlinear
  Analysis: Hybrid Systems 33 (2019) 322--335.

\bibitem{scott2014input}
J.~K. Scott, R.~Findeisen, R.~D. Braatz, D.~M. Raimondo, Input design for
  guaranteed fault diagnosis using zonotopes, Automatica 50~(6) (2014)
  1580--1589.

\bibitem{marseglia2017active}
G.~R. Marseglia, D.~M. Raimondo, Active fault diagnosis: A multi-parametric
  approach, Automatica 79 (2017) 223--230.

\bibitem{harirchi2018guaranteed}
F.~Harirchi, N.~Ozay, Guaranteed model-based fault detection in cyber--physical
  systems: A model invalidation approach, Automatica 93 (2018) 476--488.

\bibitem{boem2018plug}
F.~Boem, S.~Riverso, G.~Ferrari-Trecate, T.~Parisini, Plug-and-play fault
  detection and isolation for large-scale nonlinear systems with stochastic
  uncertainties, IEEE Transactions on Automatic Control 64~(1) (2018) 4--19.

\bibitem{chang2013new}
X.~Chang, G.~Yang, New results on output feedback ${H}_{\infty}$ control for
  linear discrete-time systems, IEEE Transactions on Automatic Control 59~(5)
  (2013) 1355--1359.

\bibitem{scherer1997multiobjective}
C.~Scherer, P.~Gahinet, M.~Chilali, Multiobjective output-feedback control via
  lmi optimization, IEEE Transactions on Automatic Control 42~(7) (1997)
  896--911.

\bibitem{pan2019static}
K.~Pan, P.~Palensky, P.~Mohajerin~Esfahani, From static to dynamic anomaly
  detection with application to power system cyber security, IEEE Transactions
  on Power Systems 35~(2) (2019) 1584--1596.

\bibitem{vershynin2018high}
R.~Vershynin, High-dimensional probability: An introduction with applications
  in data science, Vol.~47, Cambridge university press, 2018.

\bibitem{rosa2011distinguishability}
P.~Rosa, C.~Silvestre, On the distinguishability of discrete linear
  time-invariant dynamic systems, in: 2011 50th IEEE Conference on Decision and
  Control and European Control Conference, IEEE, 2011, pp. 3356--3361.

\bibitem{de2002extended}
M.~C. De~Oliveira, J.~C. Geromel, J.~Bernussou, Extended ${H}_2$ and
  ${H}_{\infty}$ norm characterizations and controller parametrizations for
  discrete-time systems, International Journal of Control 75~(9) (2002)
  666--679.

\bibitem{meyer2000matrix}
C.~D. Meyer, Matrix analysis and applied linear algebra, Vol.~71, Siam, 2000.

\bibitem{van2022multiple}
C.~Van~der Ploeg, M.~Alirezaei, N.~Van De~Wouw, P.~Mohajerin~Esfahani, Multiple
  faults estimation in dynamical systems: Tractable design and performance
  bounds, IEEE Transactions on Automatic Control (2022).

\bibitem{2004YALMIP}
J.~Lfberg, Yalmip : A toolbox for modeling and optimization in matlab, in:
  Proceedings of the CACSD Conference, 2004.

\bibitem{2018Guaranteed}
F.~Harirchi, N.~Ozay, Guaranteed model-based fault detection in
  cyber–physical systems: A model invalidation approach, Automatica 93 (2018)
  476--488.

\end{thebibliography}

\end{document}